\documentclass{acmart}
    \usepackage{multicol}
    \usepackage[shortlabels]{enumitem}
    \usepackage{bussproofs}

    \theoremstyle{definition}
    \newtheorem{theorem}{Theorem}
    \newtheorem{lemma}[theorem]{Lemma}
    \newtheorem{coro}[theorem]{Corollary}

    \newcommand{\Meta}[1]{\mathsf{#1}} 
    \newcommand{\Op}[1]{\mathtt{#1}}

    \newcommand{\MM}{\mathbb{M}}
    \newcommand{\otto}{\leftrightarrow}
    \newcommand{\At}{\Meta{At}}
    \newcommand{\XX}{\Op{X}}
    \newcommand{\Dd}{\Op{D}}
    \newcommand{\UU}{\mathrel{\Op{U}}}
    \newcommand{\FF}{\Op{F}}
    \newcommand{\GG}{\Op{G}}
    \newcommand{\LL}{\mathcal{L}}
    \newcommand{\ooo}{\Op{ooo}}
    \newcommand{\loss}{\Op{loss}}
    \newcommand{\bloss}{\Op{bloss}}
    \newcommand{\obs}{\Op{obs}}
    \newcommand{\cor}{\Op{cor}}
    \newcommand{\st}{\Op{st}}
    \newcommand{\dd}{\Meta{d}}
    \newcommand{\LTL}{\mathrm{LTL}}
    \newcommand{\EDMon}{\mathrm{EDMon}}
    \newcommand{\Lift}{\Meta{Lift}}
    \newcommand{\shift}{\Meta{shift}}
    \newcommand{\head}{\Meta{head}}
    \newcommand{\id}{\Meta{id}}
    \newcommand{\Pp}{\Op{P}}
    \newcommand{\init}{\Op{init}}
    \newcommand{\trace}{\Meta{trace}}
    \newcommand{\set}[1]{\{#1\}}
    \newcommand{\ev}{\Meta{ev}}
    \newcommand{\smon}{\mathrm{s}}
    \newcommand{\cmon}{\mathrm{c}}
    
    \newcommand{\AFR}{\Meta{AFR}}
    \newcommand{\AFS}{\Meta{AFS}}
    \newcommand{\SFR}{\Meta{SFR}}
    \newcommand{\SFS}{\Meta{SFS}}
    \newcommand{\NFR}{\Meta{NFR}}
    \newcommand{\NFS}{\Meta{NFS}}
    \newcommand{\Prox}{\Meta{Prox}}
    \newcommand{\Tolr}{\Meta{Tolr}}
    \newcommand{\Perm}{\Meta{Perm}}
    \newcommand{\Incl}{\Meta{Incl}}
    \newcommand{\Excl}{\Meta{Excl}}


\begin{document}

    \title{A Formal Framework for Noisy Runtime Verification}

    \author{Thomas Macaulay Ferguson}
    \affiliation{
        \institution{Departments of Cognitive Science and Computer Science, Rensselaer Polytechnic Institute}
        \city{Troy, NY}
        \country{USA}
        }
    \email{tfergut@rpi.edu}

    \author{Shay Allen Logan}
    \affiliation{
        \institution{Department of Philosophy, Kansas State University}
        \city{Manhattan, KS}
        \country{USA}
        }
    \email{salogan@ksu.edu}

    \author{Shawn Standefer}
    \affiliation{
        \institution{Graduate School of Information Sciences, Tohoku University}
        \city{Sendai}
        \country{Japan}
        }
    \email{standefer@tohoku.ac.jp}

     \begin{abstract}
        We introduce the logic EDMon---an epistemic dynamic logic meant to model monitorability concepts in noisy runtime verification. Its syntax and semantics are defined, and explained and the connection between EDMon and monitorability and noisy runtime verification concepts is explored. We then demonstrate that EDMon is sufficient to capture many of the results in the noisy runtime verification literature and catalog its relation to nearby logics and describe a large class of its theorems.
     \end{abstract}

     \maketitle

    \section{Introduction}

    A central aim in runtime verification (RV; see \cite{havelund2005verify,pnueli2006psl,bauer2011runtime}) is building monitors to track whether a system is entering an unsafe state. To do this, we first specify---usually in a temporal logic---properties that are associated with unsafe states. These are compiled into observers that monitor the state of the system as it evolves. When the observer can verify that the formally specified trigger-condition corresponding to the monitored property holds, it signals to an appropriate subsystem that it should intervene.\footnote{Note that this marks a polarity choice we will stick with throughout the paper: the things being monitored for are \emph{un}safe states, not \emph{safe} states. The reverse convention is fine too, but having a particular convention in play simplifies the discussion.}

    This necessarily-sketchy presentation elides a great many details. The important thing to observe is that this process can \emph{only be done at all} if the chosen property can in fact be monitored. But it's not at all obvious, in advance, what can be monitored. It turns out not to be obvious after the fact either, as there are a number of competing definitions of monitorability on offer in the literature. 

    Given this, one would like a framework for comparing different notions of monitorability. In particular, one would like a way of stating, in a rigorous, mathematically precise way, exactly how different notions differ, and what different notions can deliver. And monitorability is not the only concept in the area that could benefit from formal precisification. A second concept that could benefit from such treatment (and which is at the heart of the project pursued here) is \emph{noise}. Information channels between systems and monitors are often subject to disruptions of various sort. For example, messages from the Mars Science Laboratory sometimes arrive out of order, and dropped telemetry in unmanned aerial systems can cause spurious readings that require hand-tuning of implemented runtime monitors \cite{kauffman2021can,edwards2013relay,cauwels2020integrating}. For these and other reasons, there is a small but growing subliterature on techniques for doing what might be called \emph{noisy} runtime verification (NRV; see e.g. \cite{peled1997stutter,basin2017runtime,kauffman2021can,taleb2023uncertainty}). 
    
    Much of the work in NRV is \emph{classificatory} in nature and the classification paradigms in use are largely motivated by two questions:
    \begin{itemize}
        \item given a class of mutations $\mu$, can we characterize the properties that are immune to $\mu$?
        \item over a channel subject to mutations $\mu$, for which properties $A$ is monitoring for $A$ a trustworthy way to check whether $A$ is true? 
    \end{itemize}

    The immunity and trustworthiness concepts employed in these questions have themselves been fleshed out in a number of seemingly-inequivalent ways. One is again left wanting a formal framework for evaluating and comparing them. This paper proposes such a framework, demonstrates its adequacy for the goal of expressing, evaluating, and comparing NRV-concepts, and shows that it can also be used to fruitfully propose novel NRV concepts and results worth further exploration.
    
    \subsection{Comparison With Prior Work}

        In EDMon, mutations acting on noisy channels are modeled by dynamic modals $\langle\mu\rangle$. EDMon's semantics captures the effects of such mutations on properties expressed in linear temporal logic (LTL). As such, nearby work can be found in the realm of epistemic logic, dynamic logic, and in LTL, as well as in work on monitorability more generally. 
        
        Foundational and central work in epistemic logic can be found in  \cite{Hintikka1962, Meyer1995-MEYELF,fagin1995reasoning, Stalnaker2006-STAOLO-2,vanDitmarsch2007-VANDEL-6}. For work particularly close to the work done here see \cite{HALPERN1989195,halpernVandermeydenVardi2004}. Our work differs in that we are not considering multi-agent systems, and so no non-trivial common knowledge. We do not consider branching time, nor do we attempt complete axiomatization or complexity analysis. Our contribution lies in the addition of distinctive dynamic elements and the observation modals.
        
        Foundational and central work in dynamic logic can be found in \cite{pratt1976semantical,fischer1979propositional,harel2000dynamic}). For work particularly close to the work done here see \cite{HENRIKSEN1999187}. Relative to work in \emph{propositional} dynamic logic, our work is more expressive in some respects (as will become immediately apparent on seeing the vocabulary of our language) but also less expressive in that we restrict the range of the dynamic atoms rather severely. Relative to first-order dynamic logic, the differences are more stark, since at the moment we have only formulated EDMon in a propositional setting.

        Foundational and central work on LTL can be found in \cite{pnueli1977temporal,demri2016temporal}. For work particularly close to the work done here see \cite{bollig2026runtime} whose epistemic framing is quite similar to ours. Our work differs from the work in the last-mentioned piece primarily (though not exclusively) by its careful attention to mutation.

    \subsection{Contributions}

        The central contribution of the paper is the introduction of the logic EDMon and our demonstration of its `NRV adequacy'. In particular, we show that EDMon can 
        \begin{itemize}
            \item Usefully formalize three of the dominant monitorability concepts, 
            \item formalize an extended family of liveness classes, 
            \item capture, formally, central results both about the monitorability concepts and about their relation to liveness classes, 
            \item capture and prove central results about immunity concepts and their relation to liveness and monitorability concepts, and finally
            \item usefully point in the direction of novel concepts and results on all of the above topics.
        \end{itemize}

    \section{Syntax and Semantics}

        In this section, we describe the syntax and semantics of EDMon. We begin with a bit of preliminary material. Proofs of the lemmas included in the preliminaries are omitted to save space. They are in all cases entirely straightforward.

        \subsection{Preliminaries}

        Given a set $S$ we write $S^*$ for the set of finite sequences of members of $S$ and $S^\omega$ for the set of infinite sequences of members of $S$. We write $\varepsilon$ for the empty sequence, which is a member of $S^*$ and for each $s\in S^*$, we write $\overline{s}$ for the member of $S^\omega$ that consists of an infinite string of $s$'s. Note that we will generally not distinguish between the one element sequence containing just $s$ (which is a member of $S^*$) and $s$ itself (which is a member of $S$).
        
        For each $s\in S^*$, we write $|s|$ for its length. Given $s\in S^*\cup S^\omega$ we write $s(i)$ for the $i$th member of $s$, $s_i$ for the prefix of $s$ ending at the $i$th term and $s^i$ for the suffix of $s$ beginning at the $i+1$st term. For $s\in S^*$, $s_i=s$ if $|s|\leq i$, $s^i=\varepsilon$ if $|s|\leq i$, and $s(i)$ is undefined if $|s|<i$. Given $s\in S^*$, $u\in S^* \cup S^\omega$, we write $s. u$ for the (finite or infinite) sequence that appends $u$ to the end of $s$ in the obvious way. The following lemma is an immediate consequence of the definitions and will be used without comment in the remainder:
        \begin{lemma}
            For all $\{u,v\}\subseteq S^*\cup S^\omega$, all of the following hold:\vspace{-3mm}
            \begin{multicols}{2}\begin{itemize}
                \item $u=u_i . u^i$, 
                \item $u_i . (u^i)_k = u_{i+k}$, 
                \item $(u^i)^k = u^{i+k}$,
                \item $(u_i)_k=u_{\min(i,k)}$, and
                \item if $|u|\geq k$ then $(u_k . v)_k=u_k$ and $(u_k . v)^k = v$.
            \end{itemize}\end{multicols}
        \end{lemma}

        We will have need below to think about sets of the form $S^* \times S^\omega$. On such a set, we define three operations:
        \begin{displaymath}
            \trace(s,\sigma) = s.\sigma
            \qquad
            \shift(s,\sigma) =
                \left\{
                    \begin{array}{ll}
                        \langle s^1, \sigma \rangle & s \neq \varepsilon \\
                        \langle \varepsilon, \sigma^1 \rangle & \text{otherwise}
                    \end{array}
                \right.
            \qquad
            \head(s,\sigma) = \trace(s,\sigma)(1)
        \end{displaymath}
        We will write $\shift^k$ for the $k$-fold iteration of $\shift$, with the convention that $\shift^0(s, \sigma)=\langle s, \sigma\rangle$. We note two lemmas about $\shift^k$.
        \begin{lemma}
            For $s\in S^*$ with $|s|=k$, $\shift^k(\varepsilon,\overline{s})=\langle\varepsilon,\overline{s}\rangle$.
        \end{lemma}
        
        \begin{lemma}\label{shifthead}
            For all $s$, all $\sigma$, and all $k\geq 0$,
            $\trace(\shift^k(s, \sigma)) = \trace(s, \sigma)^k$.
        \end{lemma}

        We understand a relation between sets $S_1$ and $S_2$ to be a subset of $S_1\times S_2$. The composition $R_1\circ R_2$ of relations $R_1$ between $S_1$ and $S_2$ and $R_2$ between $S_2$ and $S_3$ is the set $\{\langle x,z\rangle:\text{ for some }y\in S_2, \langle x,y\rangle\in R_1\text{ and }\langle y,z\rangle\in R_2\}$. By a relation on $S$, we mean a relation between $S$ and $S$. Where $R$ is a relation on $S$, we write $R^k$ for the $k$-fold composition of $R$ with itself. In this same situation, we take $R^0$ to be the identity relation $\id_S$, which is just $\{\langle s,s\rangle:s\in S\}$. We will write $\nabla_S$ for the universal relation on $S$, which is to say, for $S \times S$ itself, and will write $R^*$ for the union of all $R^k$ for $0\leq k<\infty$. Given a relation $R$ between $S$ and $T$, we write $R^{-1}$ for its inverse---explicitly, for the relation between $T$ and $S$ defined by $\langle x,y\rangle\in R^{-1}$ iff $\langle y,x\rangle\in R$. When $R$ is a relation between $S$ and $T$ and $s\in S$, we write $R(s)$ for $\{t\in T:\langle s,t\rangle\in R\}$. As a bit of syntactic sugar, when given a relation $R$, we will often write $x\sim_R y$ instead of $\langle x,y\rangle\in R$. If $R$ is a relation on $X \times Y$, then each member of $R$ is a pair of pairs of the form $\langle \langle x_1, y_1 \rangle, \langle x_2, y_2\rangle \rangle$ with $\{ x_1, x_2 \} \subseteq X$ and $\{ y_1, y_2 \} \subseteq Y$. To avoid the proliferation of brackets this makes visible, we will abuse notation and write the members of relations on $X \times Y$ as four-tuples $\langle x_1, y_1, x_2, y_2 \rangle$ and identify these with pairs of pairs in the obvious way. Given relations $R_1$ between $S_1$ and $S_2$ and $R_2$ between $T_1$ and $T_2$, we write $R_1 \times R_2$ for the relation between $S_1\times T_1$ and $S_2\times T_2$ defined by $\langle s_1, t_1 \rangle \sim_{R_1 \times R_2} \langle s_2, t_2 \rangle$ iff $s_1 \sim_{R_1} s_2 $ and $t_1 \sim_{R_2} t_2$. The following standard result will be of use below:
        \begin{lemma}\label{distLem}
            $(R_1 \times R_2) \circ (R_3 \times R_4) = (R_1 \circ R_3) \times (R_2 \circ R_4)$
        \end{lemma}

    \subsection{Syntax}

        Let $\At$ be a countably infinite set of atomic formulas. For each $S\subseteq\At$ we define a language $\LL^S_{\EDMon}$ that extends the $S$-fragment of the language of LTL, ($\LL^S_{\LTL}$) with \underline{E}pistemic machinery, \underline{D}ynamic machinery, and \underline{Mon}itorability machinery; hence EDMon. On the epistemic front we add a unary `determines' modal $\Dd$. On the dynamic front we add a family of dynamic mutation modals $\langle\mu\rangle$. On the monitorability front we add observation modals $\langle\obs\rangle$ and $\langle\obs^*\rangle$. The intuitive interpretation of these operators will be described below after our discussion of the semantics.
        
        $\LL^S_{\EDMon}$ has a stratified grammar. In Backus-Naur form (BNF) we need four syntactic classes to describe $\LL^S_{\EDMon}$: LTL-formulas $\phi$, mutations $\mu$, observations $o$, and EDMon-formulas $A$. Formally, the language is as follows:
        \begin{align*}
            \phi & := p\in S\mid \neg\phi\mid \phi\lor\phi\mid \XX\phi \mid \phi\UU\phi \\
            \mu & := \st \mid \loss \mid \ooo \mid \cor \mid \mu \cup \mu \mid \mu; \mu \mid \mu^- \mid \mu^* \\
            o & :=  \obs\mid o^*\\
            A & := p\in S\mid \neg A\mid A\lor A \mid \Dd A\mid \langle\mu\rangle A\mid\langle o\rangle A\mid \XX \phi \mid \phi \UU \phi
        \end{align*}
        The top line of the definition gives us that $\LL^S_{\LTL} \subset \LL^S_{\EDMon}$. If $\phi\in\LL_{\LTL}$ contains no occurrences of $\XX$ or $\UU$, then we say that $\phi$ is \emph{propositional}. In the remainder, we allow ourselves to write `$\LL_{\EDMon}$' when either $S=\At$ or $S$ doesn't matter. So we will refer to $2^\At$ as $\Sigma$, to $x\in \Sigma$ as a \emph{state}, etc. Also, it's helpful to observe that (due to the stratification of the language) dynamic, observation, and epistemic modals \emph{cannot occur} in the scope of temporal modals. So while $\Dd \XX p$ is well-formed; $\XX \Dd p$ is not. As a pronunciation guide we note that `$\XX$' is read as `next', `$\UU$' as `until', `$\Dd$ as `determines' `$\st$' as `stutter', `$\loss$' as `loss', `$\ooo$' as `out of order', $\cor$ as `corruption', and `$\obs$' as `observe'. 
        
        As convenient abbreviations we note that $A \land B$ abbreviates $\neg (\neg A \lor \neg B)$, that $A \to B$ abbreviates $\neg A \lor B$, that $\top$ abbreviates $p\lor\neg p$, that $\FF\phi$ (``eventually $\phi$'') abbreviates $\top\UU\phi$, that $\GG\phi$ (``always $\phi$'') abbreviates $\neg\FF\neg\phi$, that $\Pp A$ (read `$A$ is permitted') abbreviates $\neg\Dd\neg A$, and that $[\mu]$ abbreviates $\neg\langle\mu\rangle\neg$, $[\obs]$ abbreviates $\neg\langle\obs\rangle\neg$, and $[\obs^*]$ abbreviates $\neg\langle\obs^*\rangle\neg$. 

        Also useful are the following recursive definitions where $\phi$ is an LTL-formula, $\mu$ is a mutation, $A$ is an EDMon-formula, and $m$ a modal, by which we mean that $m\in\{\Dd, \Pp, \langle\obs\rangle, [\obs]\} \cup \{\langle\mu\rangle, [\mu] : \mu$ a mutation$\}$:
        \begin{multicols}{2}\begin{itemize}
            \item $\XX^0\phi := \phi$; $\XX^{n+1}\phi := \XX \XX^n\phi$.
            \item $\FF^{\geq 1}\phi:=\FF \phi$; $\FF^{\geq n+1}\phi:=\FF(\phi\land\XX\FF^{\geq n}\phi)$
            \item $\mu^1 := \mu$; $\mu^{k+1} := \mu;\mu^k$; $\langle\mu^0\rangle A := A$.
            \item $m^0 A := A$; $m^{k+1} A := mm^k A$.
        \end{itemize}\end{multicols}
        We read $\XX^n\phi$ as `in $n$ ticks, $\phi$' and $\FF^{\geq n}\phi$ as `$\phi$ will happen at least $n$ times'.

    \subsection{Semantics}

        Given $S\subseteq\At$, we will refer to each $x\in 2^S=\Sigma_S$ as an $S$-\emph{state}. An $S$-\emph{trace} is a member of $\Sigma_S^\omega$. A set of $S$-traces is called an $S$-property. For each $S\subseteq T\subseteq\At$ we can interpret $\LL^S_{\EDMon}$ in a structure $\MM_T$ whose carrier set is $\Sigma_T^* \times \Sigma_T^\omega$. In particular, then, for all $S\subseteq\At$, we can interpret $\LL^S_{\EDMon}$ in a structure whose carrier set is $\Sigma^*\times\Sigma^\omega$, and this is what we will typically do. Members of any of these structures will be called \emph{points}. An \emph{initial point} is a point of the form $\langle\varepsilon,\sigma\rangle$. We will use $\vDash$ for the forcing relation that holds between points and $\LL^S_{\EDMon}$-formulas, and we will write $\llbracket -\rrbracket$ for the function assigning a relation on $\Sigma_S^* \times \Sigma_S^\omega$ to each mutation $\mu$ and each observation $o$. In order to describe $\vDash$ and $\llbracket-\rrbracket$ perspicuously, it helps to first define the following relations on $\Sigma_S^*$:

        \begin{itemize}[leftmargin = *] 
        \setlength\itemsep{1ex}
            \item $\widehat{\ooo} := \{\langle s, t \rangle \mid s = t$ or for some $\{u,v\} \subseteq \Sigma_S^*$ and $\{x,y\} \subseteq \Sigma_S$, $s = uxyv$ and $t = uyxv\}$.
            \item $\widehat{\loss} := \{\langle s, t \rangle \mid s = t$ or for some $\{u,v\} \subseteq \Sigma_S^*$ and $x\in\Sigma_S$, $s = uxv$ and $t = uv\}$.
            \item $\widehat{\cor} := \{\langle s, t \rangle \mid s = t$ or for some $\{u,v\} \subseteq \Sigma_S^*$ and $\{x,y\} \subseteq \Sigma_S$, $s = uxv$ and $t = uyv\}$.
            \item $\widehat{\st} := \{\langle s, t \rangle \mid s = t$ or for some $\{u,v\} \subseteq \Sigma_S^*$ and $x \in \Sigma_S$, $s = uxv$ and $t = uxxv\}$.
        \end{itemize}
        These relations are slightly more general than the corresponding relations studied in other papers on NRV. For example, compared to the relations in \cite{lozes2011reliable} and \cite{kauffman2021can}, ours differ by allowing $\langle\varepsilon, \varepsilon\rangle$ in each relation. That said, the differences are, on the whole, quite minor.
        
        Returning to the thread, it's helpful to define a function $\Lift$ that maps relations on $\Sigma_S^*$ to relations on $\Sigma_S^* \times \Sigma_S^\omega$ by the following rule:
        \begin{displaymath}
            \Lift(R) = R \times \id_{\Sigma_S^\omega} = \{ \langle s, \sigma, t, \tau \rangle \mid s \sim_R t \text{ and } \sigma = \tau \}
        \end{displaymath}
        It turns out that $\Lift$ is a homomorphism from the algebra of relations on $\Sigma_S^*$ to the algebra of relations on $\Sigma_S^* \times \Sigma_S^\omega$ in the sense described in the following lemma:
        \begin{lemma}\label{RSLem}
            If $R$, $T$, and $R_i$ for $i\in I$ are all relations on $\Sigma_S^*$, then all of the following are true:
            \begin{multicols}{3}\begin{enumerate}[leftmargin=*]
                \item If $R \subseteq T$, then $\Lift(R) \subseteq \Lift(T)$
                \item $\Lift(\id_{\Sigma_S^*}) = \id_{\Sigma_S^* \times \Sigma_S^\omega}$
                \item $\Lift(\bigcup_{i\in I} R_i) = \bigcup_{i\in I} \Lift(R_i)$
                \item $\Lift(R \circ T) = \Lift(R) \circ \Lift(T)$
                \item $\Lift(R^{-1}) = \Lift(R)^{-1}$
                \item $\Lift(R^*) = \Lift(R)^*$
            \end{enumerate}\end{multicols}
        \end{lemma}
        \begin{proof}
            For (1), observe that if $R\subseteq T$, then $\Lift(R) = R \times \id_{\sigma^\omega_S} \subseteq T \times \id_{\sigma^\omega_S} = \Lift(T)$. (2) is immediate. For (3), observe that $\Lift(\bigcup_{i \in I} R_i) = (\bigcup_{i \in I} R_i) \times \id_{\Sigma^\omega_S} = \bigcup_{i \in I} (R_i \times \id_{\Sigma^\omega_S}) = \bigcup_{i \in I}\Lift(R_i)$.

            For (4), observe that $\langle s, \sigma, t, \tau\rangle \in \Lift(R \circ T)$ iff $\sigma = \tau$ and there is $u$ so that $s \sim_R u$ and $u \sim_T t$. But this happens iff $\sigma = \tau$ and there is $u$ so that $\langle s, \sigma, u, \sigma\rangle \in \Lift(R)$ and $\langle u, \sigma, t, \tau\rangle \in \Lift(T)$, which happens iff $\langle s, \sigma, t, \tau\rangle \in \Lift(R) \circ \Lift(T)$. 

            For (5), observe that $\langle s, \sigma, t, \tau\rangle \in \Lift(R^{-1})$ iff $\tau=\sigma$ and $t \sim_R s$ iff $\langle t, \tau, s, \sigma\rangle \in \Lift(R)$ iff $\langle s, \sigma, t, \tau\rangle \in \Lift(R)^{-1}$. Finally, note that (6) follows immediately from (3) and (4). 
        \end{proof}
        
        \noindent To interpret the epistemic and observation modals, we use the following relations:
        
        \begin{itemize}[leftmargin=*]\setlength\itemsep{1ex}
            \item $\dd := \{ \langle s, \sigma, t, \tau \rangle : s = t \} = \id_{\Sigma_S^*} \times \nabla_{\Sigma_S^\omega}$
            \item $\llbracket \obs \rrbracket := \{ \langle s, \sigma, s . \sigma_1, \sigma^1 \rangle : \langle s, \sigma \rangle \in \Sigma_S^* \times \Sigma_S^\omega \}$.
        \end{itemize} 
        
        Where $\alpha$ and $\beta$ are mutations, $\theta$ is either a mutation or an observation, $p\in S\subseteq\At$, $\{\phi,\psi\}\subseteq\LL^S_{\LTL}$ and $\{A,B\}\subseteq\LL^S_{\EDMon}$, we define $\llbracket-\rrbracket$ and $\vDash$ as follows:
        \begin{multicols}{2}\begin{itemize}[leftmargin=*]\setlength\itemsep{1ex}
            \item $\llbracket\mu\rrbracket = \Lift(\widehat{\mu})$ for $\mu$ an atomic mutation.
            \item $\llbracket\alpha;\beta\rrbracket=\llbracket\alpha\rrbracket\circ\llbracket\beta\rrbracket$.
            \item $\llbracket\alpha\cup\beta\rrbracket=\llbracket\alpha\rrbracket\cup\llbracket\beta\rrbracket$.
            \item $\llbracket\theta^*\rrbracket=\llbracket\theta\rrbracket^*$.
            \item $\llbracket\alpha^-\rrbracket =\llbracket\alpha\rrbracket^{-1}$.
            \item $s, \sigma \vDash p$ iff $p \in \head(s, \sigma)$ for $p \in S$.
            \item $s, \sigma \vDash \neg A$ iff $s, \sigma \not\vDash A$. 
            \item $s, \sigma \vDash A \lor B$ iff $s, \sigma \vDash A$ or $s, \sigma \vDash B$.
            \item $s, \sigma \vDash \XX \phi$ iff $\shift(s, \sigma) \vDash \phi$.
            \item $s, \sigma \vDash \phi \UU \psi$ iff for some $k \geq 0$, $\shift^k(s, \sigma) \vDash \psi$ and for $0 \leq j < k$, $\shift^j(s, \sigma) \vDash \phi$. 
            \item $s, \sigma \vDash \langle \theta \rangle A$ iff $t, \tau \vDash A$ for some $\langle s, \sigma \rangle \sim_{\llbracket\theta\rrbracket} \langle t, \tau \rangle$.
            \item $s, \sigma \vDash \Dd A$ iff $t, \tau \vDash A$ for all $\langle s, \sigma \rangle \sim_{\dd} \langle t, \tau \rangle$.
        \end{itemize}\end{multicols}
        We say that $A$ is true at $\langle s,\sigma\rangle$ when $s,\sigma\vDash A$. We say $A$ is EDMon-satisfiable when $A$ is true at some point. We say that $A$ is $S$-valid just if $s,\sigma\vDash A$ for all $\langle s,\sigma\rangle\in\Sigma_S^*\times\Sigma_S^\omega$. We say that $A$ is EDMon-valid when $A$ is $\At$-valid. 

        As an initial exercise in using these clauses, we will prove that our reading of `$\FF^{\geq n}\phi$' as `$\phi$ will happen at least $n$ times' is justified:
        \begin{lemma}\label{FLemma}
        $s,\sigma\vDash\FF^{\geq n}\phi$ iff there are $n$ distinct numbers $k_1,\dots k_n$ so that for $1\leq i\leq n$, $\shift^{k_i}(s,\sigma)\vDash\phi$
        \end{lemma}
        \begin{proof}
            By induction on $n$. For $n=1$, recall $\FF^{\geq 1}\phi:=\FF\phi:=\top\UU\phi$. Now note that $s,\sigma\vDash\top\UU\phi$ iff for some $k\geq 0$, $\shift^k(s,\sigma)\vDash\phi$ and for $0\leq j<k$, $\shift^j(s,\sigma)\vDash\top$. But since all points verify $\top$, this is equivalent to saying that $s,\sigma\vDash\FF^{\geq 1}\phi$ iff for some $k\geq 0$, $\shift^k(s,\sigma)\vDash\phi$. This completes the base case.

            For the inductive step, recall that $\FF^{\geq n+1}\phi:=\FF(\phi\land\XX\FF^{\geq n}\phi)$. Now suppose that $s,\sigma\vDash\FF^{\geq n+1}\phi$. Then by similar reasoning to the previous case, for some $k\geq 0$, $\shift^k(s,\sigma)\vDash\phi\land\XX\FF^{\geq n}\phi$. So $\shift^k(s,\sigma)\vDash\phi$ and $\shift^k(s,\sigma)\vDash\XX\FF^{\geq n}\phi$. Thus, $\shift^{k+1}=\shift(\shift^k(s,\sigma))\vDash\FF^{\geq n}\phi$. So by the inductive hypothesis, there are $n$ distinct numbers $k_1,\dots,k_n$ so that for $1\leq i\leq n$, $\shift^{k_i}(\shift^{k+1}(s, \sigma)) \vDash \phi$. Thus there are $n+1$ numbers $j_1, \dots ,j_{n+1}$ with $j_1=k$ and for $i>1$, $j_i = k_{i-1}+k+1$ so that for $1 \leq i \leq n+1$, $\shift^{j_i}(s, \sigma) \vDash \phi$. Finally, note that for $i>1$, each $j_i > k = j_1$, and that this and the fact that the $k_i$ are distinct guarantees that the numbers $j_1, \dots j_{n+1}$ are all distinct.

            For the other direction, suppose there are $n+1$ distinct numbers $k_1, \dots, k_{n+1}$ so that for $1 \leq i \leq n+1$, $\shift^{k_i}(s, \sigma) \vDash \phi$. Suppose without loss of generality that $k_1 < k_2 < \dots < k_{n+1}$. Then $\shift^{k_1}(s, \sigma) \vDash \phi$ and since $k_1 < k_2$, $k_1+1\leq k_2$. So there are $n$ distinct numbers $j_1 = k_2-(k_1+1), \dots, j_n = k_{n+1}-(k_1+1)$ so that for $1 \leq i \leq n$, $\shift^{j_i}(\shift(\shift^{k_1}(s, \sigma))) \vDash \phi$. Note that since $k_1+1\leq k_2$, each $j_i$ is positive, and since the $k$'s are all distinct, so are the $j$'s. Thus by the inductive hypothesis, $\shift(\shift^{k_1}(s, \sigma)) \vDash \FF^{\geq n} \phi$. So $\shift^{k_1}(s, \sigma) \vDash \XX \FF^{\geq n} \phi$, and thus $s, \sigma \vDash \FF (\phi \land \XX \FF^{\geq n} \phi)=\FF^{\geq n+1}\phi$. 
        \end{proof}

        As a second exercise, we prove a lemma that we will need below:
        \begin{lemma}\label{obsklemma}
            $s, \sigma \vDash \langle\obs\rangle^k A$ iff $s . \sigma_k, \sigma^k \vDash A$. 
        \end{lemma}
        \begin{proof}
            By induction on $k$. If $k=0$, the result is immediate. For $k\geq 0$, $s, \sigma \vDash \langle\obs\rangle^{k+1} A$ iff $s . \sigma_1, \sigma^1 \vDash \langle\obs\rangle^k A$ iff, by the inductive hypothesis, $s . \sigma_1 . (\sigma^1)_k, (\sigma^1)^k \vDash A$. But this happens iff $s . \sigma_{k+1}, \sigma^{k+1} \vDash A$.
        \end{proof}

        Recall that the usual semantics for LTL is given in terms of infinite sequences $\sigma\in\Sigma_S^\omega$ as follows:
        \begin{multicols}{2}\begin{itemize}[leftmargin=*]\setlength\itemsep{1ex}
            \item $\sigma \vDash_{\LTL} p$ iff $p \in \sigma(1)$ for $p \in S$.
            \item $\sigma \vDash_{\LTL} \neg A$ iff $\sigma \not\vDash_{\LTL} A$. 
            \item $\sigma \vDash_{\LTL} A \lor B$ iff $\sigma \vDash_{\LTL} A$ or $\sigma \vDash_{\LTL} B$.
            \item $\sigma \vDash_{\LTL} \XX \phi$ iff $\sigma^1 \vDash_{\LTL} \phi$.
            \item $\sigma \vDash_{\LTL} \phi \UU \psi$ iff for some $k \geq 0$, $\sigma^k \vDash_{\LTL} \psi$ and for $0 \leq j < k$, $\sigma^j \vDash_{\LTL} \phi$. 
        \end{itemize}\end{multicols}
        
        As you might suspect from examining $\LL_{\EDMon}$'s semantics, when restricted to $\LL_{\LTL}$, it is essentially exactly the standard semantics: 
        \begin{lemma}\label{LTLSat}
            For $\phi \in \LL_{\LTL}$, $\trace(s, \sigma) \vDash_{\LTL} \phi$ iff $s, \sigma \vDash \phi$.
        \end{lemma}
        \begin{proof}
            By induction on $\phi$. If $\phi=p$ is an atom, then first note that $\trace(s,\sigma)(1)=\head(s,\sigma)$. Now observe that $\trace(s,\sigma)\vDash_{\LTL} p$ iff $p\in\trace(s,\sigma)(1)$ iff $p\in\head(s,\sigma)$ iff $s,\sigma\vDash p$.

            The negation and disjunction cases are immediate from the inductive hypotheses. 

            For the $\XX$-case, first note that by Lemma~\ref{shifthead}, $\trace(s, \sigma)^1 = \trace(\shift(s, \sigma))$. Now observe that $\trace(s, \sigma) \vDash_{\LTL} \XX \phi$ iff $\trace(s, \sigma)^1 \vDash_{\LTL} \phi$ iff $\trace(\shift(s, \sigma)) \vDash_{\LTL} \phi$ iff (by the inductive hypothesis) $\shift(s, \sigma) \vDash \phi$ iff $s, \sigma \vDash \XX \phi$.

            For the $\UU$-case, first note that by Lemma~\ref{shifthead}, for all $n\geq 0$, $\trace(s, \sigma)^n = \trace(\shift^n(s, \sigma))$. Now observe that $\trace(s, \sigma) \vDash \phi \UU \psi$ iff for some $k \geq 0$, $\trace(s, \sigma)^k \vDash_{\LTL} \psi$ and for all $0 \leq j < k$, $\trace(s, \sigma)^j \vDash_{\LTL} \phi$. But this happens iff for some $k\geq 0$, $\trace(\shift^k(s, \sigma)) \vDash_{\LTL} \psi$ and for all $0 \leq j < k$, $\trace(\shift^j(s, \sigma)) \vDash_{\LTL} \phi$. But by the inductive hypothesis, this happens iff for some $k\geq 0$, $\shift^k(s, \sigma) \vDash \psi$ and for all $0 \leq j < k$, $\shift^j(s, \sigma) \vDash \phi$ which in turns happens iff $s,\sigma \vDash \phi \UU \psi$.
        \end{proof}

        In plain English, what Lemma~\ref{LTLSat} says is that the location of the split between (finite) prefix and (infinite) suffix plays no role whatsoever in evaluating LTL-formulas. All that matters is the combined trace $s . \sigma$. So when it comes to evaluating LTL-formulas, the semantics treats points exactly like the combined traces they represent, as it should.
        
        In keeping with our nonproliferation policies, we will abuse notation and write `$s, \sigma \sim_\mu t, \tau$' when we should technically write `$\langle s, \sigma \rangle \sim_{\llbracket\mu\rrbracket} \langle t, \tau \rangle$'. Finally, note that the `hat' function mapping atomic mutations $\mu$ to relations $\widehat{\mu}$ on $\Sigma_S^*$ can clearly be extended to arbitrary mutations using the recursive procedure given in the semantics. Explicitly, we can define $\widehat{\mu_1; \mu_2} = \widehat{\mu_1} \circ \widehat{\mu_2}$, $\widehat{\mu_1 \cup \mu_2} = \widehat{\mu_1} \cup \widehat{\mu_2}$, $\widehat{\mu^*} = \widehat{\mu}^*$, and $\widehat{\mu^-} = \widehat{\mu}^{-1}$. As a technical matter, we also need to specify that $\widehat{\mu^0}$ is the identity relation on $\Sigma^*_S$. We then have the following results
        \begin{coro}\label{liftingcoro}
            For all mutations $\mu$, $\llbracket \mu \rrbracket = \Lift(\widehat{\mu})$. 
        \end{coro}   
        \begin{proof}
            By induction on $\mu$. If $\mu$ is atomic, the result is immediate and for every other case, the result is an immediate consequence of the inductive hypothesis and Lemma~\ref{RSLem}.
        \end{proof}

        \begin{coro}\label{easymutlem}
            $s, \sigma \vDash \langle\mu\rangle A$ iff $t, \sigma \vDash A$ for some $s \sim_{\widehat{\mu}} t$. 
        \end{coro}
        \begin{proof}
            $s, \sigma \vDash \langle \mu \rangle A$ iff $t, \tau \vDash A$ for some $s, \sigma \sim_{\mu} t, \tau$. But since $\llbracket \mu \rrbracket = \Lift(\widehat{\mu})$, $s, \sigma \sim_{\mu} t, \tau$ iff $s \sim_{\widehat{\mu}} t$ and $\sigma = \tau$. Thus, $t, \tau \vDash A$ for some $s, \sigma \sim_{\mu} t, \tau$ iff $t, \sigma \vDash A$ for some $s \sim_{\widehat{\mu}} t$. 
        \end{proof}

        A similar result to Corollary~\ref{easymutlem} is the following, which follows immediately from the definition of $\dd$:
        \begin{lemma}\label{DandPLemma}
            $s, \sigma \vDash \Dd A$ iff $s, \tau \vDash A$ for all $\tau \in \Sigma^\omega$ and $s, \sigma \vDash \Pp A$ iff $s, \tau \vDash A$ for some $\tau \in \Sigma^\omega$. 
        \end{lemma}

        \begin{lemma}\label{concatLemma}
            For all mutations $\mu$ and all $r \in \Sigma^*$, if $s \sim_{\widehat{\mu}} t$, then $s . r \sim_{\widehat{\mu}} t . r$.
        \end{lemma}
        \begin{proof}
            By induction on $\mu$. For each of the atomic cases the result is immediate on inspection. For the semicolon case, suppose $s \sim_{\widehat{\mu_1; \mu_2}} t$. Then for some $u$, $s \sim_{\widehat{\mu_1}} u$ and $u \sim_{\widehat{\mu_2}} t$. So by the inductive hypothesis, $s . r \sim_{\widehat{\mu_1}} u . r$ and $u . r \sim_{\widehat{\mu_2}} t . r$. Thus $s . r \sim_{\widehat{\mu_1; \mu_2}} t . r$ as required.

            The choice case is immediate from the inductive hypothesis. For the star case, suppose $s \sim_{\widehat{\mu^*}} t$. Then for some $k$, there are $u_1, \dots, u_k$ so that $s \sim_{\widehat{\mu}} u_1$ and for $1\leq i\leq k-1$, $u_i \sim_{\widehat{\mu}} u_{i+1}$ and $u_k \sim_{\widehat{\mu}} t$. But then by the inductive hypothesis, $s . r  \sim_{\widehat{\mu}} u_1 . r$ and for $1\leq i\leq k-1$, $u_i . r  \sim_{\widehat{\mu}} u_{i+1} . r$ and $u_k . r \sim_{\widehat{\mu}} t . r$. So $s . r \sim_{\widehat{\mu^*}} t . r$. 

            For the inverse case, like the choice case, is immediate from the inductive hypothesis. 
        \end{proof}

    \subsection{Understanding the Language}
        
        We understand points $\langle s,\sigma\rangle$ as follows: $\trace(s,\sigma)$ is the sequence containing \emph{all} the data some sensor will produce and transmit---intuitively, to a monitor---over its lifetime; $s$ is the part of $\trace(s,\sigma)$ that has been transmitted so far. The two important details here are that (a) every point describes information that is being \emph{sent}, and (b) every point is, in a certain minimal sense, temporally located. To see these two ideas in action, it helps to have a look at how the modals behave. We take each class in turn, along the way proving a number of sanity check lemmas showing that the semantic concepts behave as they intuitively ought.
        
        \subsubsection{Observation modals} Parsing semantic clauses gives that $s, \sigma \vDash \langle \obs \rangle A$ iff $s . \sigma_1, \sigma^1 \vDash A$. Note that $\trace(s . \sigma_1, \sigma^1)=\trace(s, \sigma)$. So while $\obs$ induces a change in which point we are evaluating formulas from, it doesn't change the trace at which the evaluation is happening---all it does is move us from a situation where $|s|$ states from $\trace(s,\sigma)$ have been transmitted to a situation where $|s|+1$ have. It follows that $\langle \obs \rangle A$ is true just if $A$ will be true after the sensor in question transmits its next observation. Building on this, $s, \sigma \vDash \langle \obs^* \rangle A$ just if for some $k\geq 0$, $s . \sigma_k, \sigma^k \vDash A$, which is to say that $\langle \obs^* \rangle A$ is just when $A$ will become true after the sensor makes (and sends) some not necessarily nonzero number of additional observations. 

        Since sending information doesn't (either intuitively or, as we just saw in Lemma~\ref{LTLSat}, in the semantics) change `worldly' information, one hopes that `worldly' formulas---which to say $\LL_{\LTL}$-formulas---will display an appropriate sort of observation-invariance. That they do is the content of our first `sanity check' lemma: 
        \begin{lemma}\label{LTLobsLemma}
            If $A\in\LL_{\LTL}$, then $s, \sigma \vDash A$ iff $s, \sigma \vDash \langle\obs^*\rangle A$. 
        \end{lemma}
        \begin{proof}
            $s, \sigma \vDash \langle\obs^*\rangle A$ iff (by Lemma~\ref{obsklemma}) for some $k\geq 0$, $s, \sigma \vDash \langle\obs\rangle^k A$ iff for some $k\geq 0$, $s . \sigma_k, \sigma^k \vDash A$. But by Lemma~\ref{LTLSat}, this happens iff $\trace(s . \sigma_k, \sigma^k) \vDash_{\LTL} A$. But notice that $\trace(s . \sigma_k, \sigma^k) = s . \sigma_k . \sigma^k = s . \sigma = \trace(s, \sigma)$. So $\trace(s . \sigma_k, \sigma^k) \vDash_{\LTL} A$ iff $\trace(s, \sigma) \vDash_{\LTL} A$ iff (again by Lemma~\ref{LTLSat}), $s, \sigma\vDash A$. 
        \end{proof}
            
        \subsubsection{Mutation modals} 
        
            For $\mu$ a mutation, $s, \sigma \vDash \langle \mu \rangle A$ just if (by Corollary~\ref{liftingcoro}) $t, \sigma\vDash A$ for some $s \sim_{\widehat{\mu}} t$. So in a situation where we are \emph{receiving} information that verifies $A$ along a channel subject to mutations $\mu$, $s$ is a description of what might have been sent. So receiving a message that `says' $B$ across a channel subject to mutations $\mu$ can be modeled as a sent-message that `says' $\langle\mu\rangle B$. In particular, when $\langle\mu\rangle \langle\obs\rangle A$ is true, we have received a message that will support $A$ after the next observation. When this is true, it will also be true that after the next observation, we will have received a message that supports $A$. Appealing to the above translation, we expect the following lemma, which is our next sanity check:
           \begin{lemma}\label{muobslemma}
                For all $A$ and $\mu$, $\langle\mu\rangle\langle \obs\rangle A \to \langle\obs\rangle\langle\mu\rangle A$ is valid.
            \end{lemma}
            \begin{proof}
                Suppose $s, \sigma \vDash \langle\mu\rangle\langle \obs\rangle A$. Then $t, \sigma \vDash \langle\obs\rangle A$ for some $s \sim_{\widehat{\mu}} t$. Thus $t . \sigma_1, \sigma^1 \vDash A$. But since $s \sim_{\widehat{\mu}} t$, it follows by Lemma~\ref{concatLemma} that $s . \sigma_1 \sim_{\widehat{\mu}} t . \sigma_1$. Thus $s . \sigma_1, \sigma^1 \vDash \langle\mu\rangle A$ from which it follows that $s, \sigma \vDash \langle\obs\rangle\langle \mu\rangle A$. 
            \end{proof}

            Note that the converse of this Lemma both should and does fail. It \emph{should} fail because it might be the case that after the next observation, we will have received a message that supports $A$, even though we have not currently received a message that \emph{tells us} that after our next observation we will have received a message that supports $A$. That it \emph{does} fail can be seen by considering the point $\langle \{p\}, \{q\} . \overline{\emptyset}\rangle$. Note that $\{p\}, \{q\} . \overline{\emptyset} \vDash \langle\obs\rangle \langle\ooo\rangle q$ because $\{q\}\{p\},\overline{\emptyset}\vDash q$, and $\{q\} \{p\} \sim_{\widehat{\ooo}} \{p\} \{q\}$. So $\{p\} \{q\}, \overline{\emptyset} \vDash \langle\ooo\rangle q$ and thus $\{p\}, \{q\} . \overline{\emptyset} \vDash \langle\obs\rangle \langle\ooo\rangle q$. But on the other hand $\{p\}, \{q\} . \overline{\emptyset} \not \vDash \langle\ooo\rangle \langle\obs\rangle q$ since (due to the fact the length of the sequence $\{p\}$ is 1) only points of the form $\{p\}, \tau$ are $\ooo$-related to $\{p\}, \{q\} . \overline{\emptyset}$ and no such point supports $\langle\obs\rangle q$. 
        
        \subsubsection{Epistemic modals} 
        
            It follows from Lemma~\ref{DandPLemma} that alternative version of the semantic clause for $\Dd$ reads ``$s, \sigma \vDash \Dd A$ iff $s, \tau \vDash A$ for all $\tau \in \Sigma_S^\omega$.'' From this we see that $\Dd A$ is true of exactly those formulas that are guaranteed to be (or, if you prefer, $\Dd$etermined to be) true by the finite signal $s$ that has been transmitted so far. Dually, $\Pp A := \neg\Dd\neg A$ is true of exactly those formulas that \emph{could be} true, given what's been sent so far.

            We will provide two different sanity checks for our epistemic modals. For the first, note that what you can determine after making an additional observation includes---but is, in general, strictly more than---what you can determine you will observe. This leads us to expect the following:
        \begin{lemma}\label{obsCommutativity}
            $\Dd\langle\obs\rangle A\to \langle\obs\rangle\Dd A$ is valid for all $A$.
        \end{lemma}
        \begin{proof}
            Suppose $s, \sigma \vDash \Dd \langle\obs\rangle A$. Then $s, \tau \vDash \langle\obs\rangle A$ for all $\tau$. So $s . \tau_1, \tau^1 \vDash A$ for all $\tau$. It follows that  $s . (\sigma_1 . \tau)_1, \sigma_1 . \tau)^1 \vDash A$ for all $\tau$. But $(\sigma_1 . \tau)_1=\sigma_1$ and $(\sigma_1 . \tau)^1 = \tau$. So $s . \sigma_1, \tau \vDash A$ for all $\tau$. Thus $s . \sigma_1, \sigma^1 \vDash \Dd A$. So $s, \sigma \vDash \langle\obs\rangle \Dd A$. 
        \end{proof}
        As mentioned, we expect that in general what we can determine after making an additional observation is \emph{strictly more} than what we can determine we will observe prior to making the observation. So we expect that the converse of Lemma~\ref{obsCommutativity} should fail. That it does can be seen by again thinking about $\langle\{p\}, \{q\} . \overline{\emptyset}\rangle$. Observe that $\{q\}, \tau = \shift(\{p\} \{q\}, \tau) \vDash q$ for all $\tau$. Thus $\{p\} \{q\}, \tau \vDash \XX q$ for all $\tau$. So $\{p\} \{q\}, \overline{\emptyset} = \{p\} . (\{q\} . \overline{\emptyset})_1, (\{q\} . \overline{\emptyset})^1 \vDash \Dd \XX q$. Thus $\{p\}, \{q\} . \overline{\emptyset} \vDash \langle\obs\rangle \Dd \XX q$---in words, if $q$ is true at the next tick, then observing the next tick will let you determine that $q$ is true at the next tick. But the reverse is of course not true---the fact that $q$ is true at the next tick doesn't mean you can determine, now, that you will observe $q$ at the next tick. And the semantics bears this out. To see this, note that $\emptyset, \overline{\emptyset} = \shift(\{p\} \emptyset, \overline{\emptyset}) \not\vDash q$. So $\{p\} \emptyset, \overline{\emptyset} = \{p\} . \overline{\emptyset}_1, \overline{\emptyset}^1 \not\vDash \XX q$. Thus $\{p\}, \overline{\emptyset} \not\vDash \langle\obs\rangle \XX q$. Thus it is not the case for all $\tau$ that $\{p\}, \tau \vDash \langle\obs\rangle \XX q$. So $\{p\}, \{q\} . \overline{\emptyset} \not\vDash \Dd \langle\obs\rangle \XX q$. 

        So epistemic and observation modalities \emph{half} commute. When it comes to mutation, things are friendlier. Mathematically, it's clear enough why this is the case: on a point $\langle s,\sigma\rangle$, mutations only interact with the $s$-coordinate and determination only interacts with the $\sigma$-coordinate. So mutation and determination modals simply don't interact which should lead us to expect they fully commute. 
        
        To see why this should \emph{intuitively} hold is a bit harder. But the difficulty of explaining this fact intuitively can be explained by the mathematical fact: since noise and determination act on different coordinates, the order in which they act is irrelevant. What records this in ordinary language is something that comes across as trite---e.g. `if, no matter how my message is mutated and no matter what the future holds, my message still says $A$, then no matter what the future holds and no matter how my message is mutated, my message still says $A$ (and vice-versa).' But triteness is what's to be expected when recording something essentially tautological, so this isn't a failure of the model or of the paraphrase; it's in fact what we should expect. At any rate, here's the lemma:
        \begin{lemma}\label{mucomm}
            For all $A$ and all mutations $\mu$, $\Dd [\mu] A \otto [\mu] \Dd A$ is valid.
        \end{lemma}
        \begin{proof}
            By Corollary~\ref{easymutlem} and Lemma~\ref{DandPLemma}, $s, \sigma \vDash \Dd [\mu] A$ iff $t, \tau \vDash A$ for all $\tau \in \Sigma^\omega$ and all $s \sim_{\widehat{\mu}} t$ which happens iff $t, \tau \vDash A$ for all $s \sim_{\widehat{\mu}} t$ and all $\tau \in \Sigma^\omega$, iff $s, \sigma \vDash [\mu] \Dd A$.
        \end{proof}

    \subsection{Whence Monitors?}

        We now turn to describing how to use $\LL_{\EDMon}$ to model the epistemic state of a monitor. There are two ways to do this, and they correspond to two different questions we might imagine a (personified) monitor asking. First question: what can I, qua monitor, conclude about the sender on the basis of what I've received? Recall that \emph{receiving} $A$ along a channel subject to mutations $\mu$ is modeled here by $\langle \mu \rangle A$. Of course, typically one receives a whole trace, not just a single formula. But we can model the former as the latter in real-world scenarios by replacing traces with their characteristic formulas. Putting the trace/formula matter aside, what we are asking when we ask this question is the following: given that $\langle \mu \rangle A$ is what we know about the state of the sent message, what is \emph{determined}? Stating it this way makes the answer clear: what's determined are those $B$ so that $\langle \mu \rangle A\to \Dd B$ is valid. Note that though this is clearly an epistemic question, it isn't really a monitoring question; monitors look for particular $A$ that will trigger an intervention. 
        
        So we move on to the second question: given we want to act when $A$ is true of the sender, what received messages should trigger our action? Again we recall that receiving $B$ along a channel subject to mutations $\mu$ is modeled by $\langle\mu\rangle B$. The monitor wants to act on all $B$ that, when received, guarantee $A$ was true of the sender. So they should act when $\langle\mu\rangle B \to A$ is valid. Let's call any such formula a $\mu$-sufficient surrogate for $A$, or (when context allows it) just a \emph{surrogate}.

        The main problem with surrogates is that there are too many of them. To see this, suppose $B$ is a surrogate for $A$ and let $C$ be any formula at all. Then $B\land C$ is (as we leave the reader to check) also a surrogate. Continuing down this road leads to the ur-problem: $\bot$ is a surrogate for absolutely every $A$.

        To start solving this problem, it helps to observe that each step down this road leaves us with a \emph{less} sensitive surrogate. Consider, for example, the very first step, where we decide to use $B\land C$ as our surrogate rather than $B$. If we need to act whenever $A$ is true and $\langle\mu\rangle B\to A$ really is valid, then anytime we receive a message that says $B$, we should act. Replacing $B$ with $B\land C$ means we'll instead act only when both $B$ and $C$ are true. In general this means we'll have decided to act \emph{less often}, which means we will (in general) more frequently fail to act in cases when we ought to. 

        So what we really want to do, qua monitor, is find the \emph{weakest} surrogate---that is, the weakest $B$ so that $\langle\mu\rangle B\to A$ is valid. But adjoint behaviors of the converse operator tell us essentially immediately what the weakest such $B$ is: if $\langle\mu\rangle B\to A$ is valid, then so is $B\to[\mu^-] A$, so $[\mu^-]A$ is implied by every surrogate for $A$. And $[\mu^-]A$ is itself a surrogate for $A$ because $\langle\mu\rangle[\mu^-]A\to A$ is also valid.\footnote{This can be seen either by thinking about the unit for the adjunction here or by straightforwardly working through the semantic clauses.} So monitoring for $[\mu^-]A$ is monitoring for the most sensitive surrogate for $A$ one can possibly monitor for, given the noise at hand. 

        The problem is that (again speaking as a monitor) this isn't a useful solution. There are two reasons for this. The first is strictly practical: existing monitors (e.g. R2U2 \cite{johannsen2023r2u2,aurandt2025towards} or RTLola \cite{baumeister2020rtlola,baumeister2024tutorial}) can't work with EDMon formulas. What they can work with are formulas in LTL-like languages. The second is more theoretical: it's not at all clear how one would monitor $[\mu^-]A$. In fact, there's a sense in which the `discovery' that $[\mu^-]A$ is the most sensitive surrogate for $A$ is less a discovery and more a formal restatement of the problem at hand---monitoring for $[\mu^-]A$ \emph{just is} checking one's message to see whether it could only have arisen, given $\mu$, from a message that verified $A$. But we already knew that that's what we wanted to do. What we don't know is what doing it entails, and $[\mu^-]A$ doesn't give us any hints.

        These problems are, it turns out, two sides of the same coin. And we can solve them (in a range of cases) simultaneously. We won't have space to provide the solution in full detail in this paper. Instead, we will simply examine one case study and leave the matter to the next paper. The point, to state it clearly, is that EDMon can play a role in providing actual, implementable solutions to monitorability problems. 

        Before getting to the case study (Theorem~\ref{FFtheorem}), we need a Lemma:
        \begin{lemma}\label{propositional_lem}
            If $A$ is propositional and $\head(s, \sigma) = \head(t, \tau)$, then $s, \sigma \vDash A$ iff $t, \tau \vDash A$. 
        \end{lemma}
        \begin{proof}
            By a straightforward induction on $A$. 
        \end{proof}
        
        \begin{theorem}\label{FFtheorem}
            For $k\geq 1$, $l\geq 0$, $A\in\LL_{\LTL}$ and $B$ propositional, if $\langle\cor^l\rangle A \to\FF^{\geq k}B $ is EDMon-valid, then $A \to \FF^{\geq k+l}B $ is too.
        \end{theorem}
        \begin{proof}
            First note that if $B$ is a propositional tautology, then $B$ is true at every state. So by Lemma~\ref{FLemma}, $\FF^{\geq k+l}B$ is valid, and thus so too is $A \to \FF^{\geq k+l}B$, verifying the theorem. So we suppose $B$ is not a propositional tautology and let $x\in\Sigma_S$ be a state that falsifies $B$. 

            Our argument is by contraposition. So, we suppose that some point falsifies $A \to \FF^{\geq k+l} B$ and construct a point that falsifies $\langle\cor^l\rangle A \to \FF^{\geq k} B$. 

            To begin let $s,\sigma \not\vDash A \to \FF^{\geq k+l} B$. Then $s, \sigma \vDash A$ and $s, \sigma \not\vDash \FF^{\geq k+l}B$. Let $M=\{m:\shift^{m-1}(s, \sigma) \vDash B\}$.\footnote{The `$-1$' might seem mysterious. It can be explained by observing that if $\shift^m(s,\sigma)\vDash B$, then the first member of $\trace(\shift^m(s, \sigma))$ is what actually verifies $B$, but the first member of $\trace(\shift^m(s, \sigma))$ is $\trace(s, \sigma)(m+1)$.} By Lemma~\ref{FLemma}, since $s, \sigma \not\vDash \FF^{\geq k+l}B$, $|M| < k+l$. Let $j=|M|-(k-1)$, and note that $j\leq l$.
            
            If $j\leq 0$, then $|M| <k$, so (again by Lemma~\ref{FLemma}), $s, \sigma \not\vDash \FF^{\geq k} B$. Since $s, \sigma \vDash A$ and $s, \sigma \sim_{\cor^l} s, \sigma$, it follows that $s, \sigma \not\vDash \langle\cor^l\rangle A \to \FF^{\geq k} B$. If $j>0$, let $m^+=\max(M)$ and $\Delta$ be a $j$-element subset of $M$. Define $\tau'$ by $\tau'(i) = \trace(s, \sigma)(i)$ for $i \not\in \Delta$ and $\tau'(i)=x$ otherwise. Let $t=\tau'_{m^+}$ and $\tau=(\tau')^{m^+}$. We claim that $t, \tau \not\vDash \langle\cor^l\rangle A \to \FF^{\geq k}B$. 

            To see this, first define $s' = \trace(s, \sigma)_{m^+}$ and $\sigma' = \trace(s, \sigma)^{m^+}$. By construction, $s'$ and $t$ differ at exactly the points in $\Delta$, which has cardinality $j \leq l$. So 
            $t, \tau \sim_{\cor^l} s', \sigma'$. Also note that $\trace(s', \sigma') = \trace(s, \sigma)_{m^+} . \trace(s, \sigma)^{m^+} = \trace(s, \sigma)$, so by Lemma~\ref{LTLSat}, since $A\in\LL_{\LTL}$ and $s, \sigma \vDash A$ we also get that $s', \sigma' \vDash A$. So $t, \tau \vDash \langle \cor^l\rangle A$. Now observe that by construction, for $m \not\in \Delta$, $\head(\shift^m(t, \tau)) = \head(\shift^m(s, \sigma))$.
            So since $B$ is propositional, for $m \not\in \Delta$ $\shift^m(t, \tau) \vDash B$
            iff $\shift^m(s, \sigma) \vDash B$. Again by construction $\shift^m(t, \tau) \not\vDash B$ for $m \in \Delta$. So there are at most $|M|-j\leq k-1$ numbers $n$ so that $\shift^n(t, \tau) \vDash B$. So $t, \tau \not\vDash \FF^{\geq k} B$ and thus $t, \tau \not\vDash \langle\cor^l\rangle A \to \FF^{\geq k}B$. 
        \end{proof}

        \begin{theorem}
             For $k\geq 1$, $l\geq 0$, $A\in\LL_{\LTL}$ and $B$ propositional, if $\langle\cor^l\rangle A \to\neg\FF^{\geq k+l}B $ is valid then $A \to\neg\FF^{\geq k}B $ is too. 
        \end{theorem}
        \begin{proof}
            First note that if $B$ is a propositional contradiction, then $B$ is false at every state. So by Lemma~\ref{FLemma}, $\neg\FF^{\geq k}B$ is valid, and thus so too is $A \to \neg\FF^{\geq k+l}B$, verifying the theorem. So we suppose $B$ is not a propositional contradiction and let $x\in\Sigma_S$ be a state that verifies $B$. 

            As before, we argue by contraposition. Let $s, \sigma \not\vDash A \to \neg \FF^{\geq k} B$. Then $s, \sigma \vDash A$ and $s, \sigma \vDash \FF^{\geq k} B$. Let $M=\{m : \shift^{m-1}(s, \sigma) \vDash B\}$. Then by Lemma~\ref{FLemma}, $|M| \geq k$. Choose a $k$-element subset, $M_k$, of $M$ and let $m^+=\max(M_k)$. We define $\tau'$ by $\tau'(i)=x$ if $m^+ < i \leq m^+ + l$, and $\tau'(i) = \trace(s, \sigma)(i)$ otherwise. Let $t = \tau'_{m^+ + l}$ and $\tau = (\tau')^{m^+ + l}$, also let $s' = \trace(s, \sigma)_{m^+ +l}$ and $\sigma' = \trace(s, \sigma)^{m^+ + l}$. 

            Since by construction $\trace(s', \sigma') = \trace(s, \sigma)$ and $A\in\LL_{\LTL}$, $s', \sigma' \vDash A$. Also note that by construction, $t,\tau\sim_{\cor^l} s',\sigma'$. So $t,\tau \vDash \langle\cor^l\rangle A$. But also note that there are at least $k+l$ numbers $n$ so that $\shift^n(t, \tau) \vDash B$: the $k$-members of $\{m-1 : m\in M_k\}$ and the $l$ numbers $m^+, \dots, m^+ + l -1$, none of which occur in the previous set since $m^+=\max(M_k)$. So by Lemma~\ref{FLemma} again, $t, \tau \vDash \FF^{\geq k+l} B$ as well, so $t, \tau \not\vDash \langle\cor^l\rangle A \to \neg \FF^{\geq k+l} B$. 
        \end{proof}

        \begin{theorem}
            For $k\geq 1$ and $l\geq 0$, $\langle\cor^l\rangle \FF^{\geq k+l} B \to \FF^{\geq k}B$ is EDMon-valid for all propositional $B$.
        \end{theorem}
        \begin{proof}
            Again by contraposition. Let $s, \sigma \not\vDash \FF^{\geq k}B$. Then by Lemma~\ref{FLemma}, there are at most $k-1$ numbers $i$ so that $\shift^i(s, \sigma)\vDash B$. Choose $t,\tau$ so that $s,\sigma\sim_{\cor^l} t,\tau$. Then there are at most $l$ numbers $i$ so that $\trace(s, \sigma)(i) \neq \trace(t, \tau)(i)$. It follows that there are at most $k-1+l < k+l$ numbers $i$ so that $\shift^i(t, \tau) \vDash B$. So $t, \tau \not\vDash \FF^{\geq k+l} B$. So since $t, \tau$ was an arbitrary point $\cor^l$-related to $s, \sigma$ we have that $s, \sigma \not\vDash \langle\cor^l\rangle \FF^{\geq k+l} B$. 
        \end{proof}

        \begin{theorem}
            For $k\geq 1$ and $l\geq 0$, $\langle\cor^l\rangle \neg \FF^{\geq k} B \to \neg \FF^{\geq k+l} B$ is EDMon-valid for all propositional $B$.
        \end{theorem}
        \begin{proof}
            Essentially the same as in the previous theorem: let $s, \sigma \vDash \FF^{\geq k+l} B$ and $s,\sigma\sim_{\cor^l} t,\tau$. Then $\trace(s,\sigma)(i)$ and $\trace(t,\tau)(i)$ differ at most $l$ times. So $t,\tau\vDash\FF^{\geq k}B$. Thus $s, \sigma \not\vDash \langle\cor^l\rangle \neg \FF^{\geq k} B$. 
        \end{proof}

        It follows that for propositional $B$, $\FF^{\geq k+l}B$ is the weakest $\cor^l$-sufficient surrogate for $\FF^{\geq k}B$ and $\neg\FF^{\geq k}B$ is the weakest $\cor^l$-sufficient surrogate for $\neg\FF^{\geq k+l}B$. Since we also know that $[(\cor^l)^-] \FF^{\geq k} B$ is the weakest $\cor^l$-sufficient surrogate for $\FF^{\geq k}B$ and $[(\cor^l)^-] \neg\FF^{\geq k+l}B$ is the weakest $\cor^l$-sufficient surrogate for $\neg\FF^{\geq k+l}B$, this also demonstrates that $[(\cor^l)^-] \FF^{\geq k} B$ is equivalent to $\FF^{\geq k+l}B$ and $[(\cor^l)^-] \neg \FF^{\geq k+l} B$ is equivalent to $\neg\FF^{\geq k}B$. So, at least some of the time, $[\mu^-]A$ can be shown to be equivalent to an LTL-formula. In such cases, we can monitor for it in existing systems in a well-defined way, and the EDMon framework lets us certify that doing so is both sound and optimal.

        \subsection{Monitoring on Multiple Channels}

        It's not uncommon (see citations) to have multiple channels monitoring the same instrument. In such cases, we could have different types of noise arise on the different monitoring channels. Traditional approaches to noisy monitoring \emph{can} make use of this. But in the present framework, we can make much better use of it.

        To say more, we introduce two abbreviations:
            \begin{itemize}
                \item $\triangleright_{i=0}^{n}\phi_i$ will abbreviate $\displaystyle\bigwedge_{i=0}^n\XX^i\phi_i$.
                \item $\widehat{\phi}^n_i$ will abbreviate $\displaystyle\bigwedge_{\substack{0\leq j\leq n\\ j\neq i}}\neg\phi_j$.
            \end{itemize}
        So, loosely, $\XX^n\phi$ means `$n$ steps from now $\phi$ will be true'; `$\triangleright_{i=0}^{n}\phi_i$' means `$\phi_0$ and then $\phi_1$ and then\dots and then $\phi_n$'; and `$\widehat{\phi}^n_i$' means `maybe $\phi_i$, but none of the other $\phi$'s'. 
        \begin{theorem}\label{multichan}
        For all $n\geq 1$, all \emph{propositional} instances of the following are EDMon-valid:
            \begin{displaymath}
                (\langle\ooo\rangle \triangleright_{i=0}^n \phi_i 
                \land 
                \langle\cor\rangle \triangleright_{i=0}^n \widehat{\phi}^n_i)
                \to \triangleright_{i=0}^{n-1}\phi_i
            \end{displaymath}
        \end{theorem}
        \begin{proof}
            We provide the proof for the $n=2$ case. The general case can be proved following similar principles. So, supposing that $n=2$, let $s,\sigma\vDash\langle\ooo\rangle\triangleright_{i=0}^2\phi_i$ and $s,\sigma\vDash\langle\cor\rangle\triangleright_{i=0}^2\widehat{\phi}^2_i$. Then there are $t$ and $r$ with $s,\sigma\sim_{\ooo}t,\sigma$ and $s,\sigma\sim_{\cor}r,\sigma$ such that $\trace(t,\sigma)(1)\vDash\phi_0$ and $\trace(t,\sigma)(2)\vDash\phi_1$ and $\trace(t,\sigma)(3)\vDash\phi_2$ and such that $\trace(r,\sigma)(1)\vDash\widehat{\phi}^2_0$ and $\trace(r,\sigma)(2)\vDash\widehat{\phi}^2_1$ and
            $\trace(r,\sigma)(3)\vDash\widehat{\phi}^2_2$.
            
            Since $s,\sigma\sim_{\ooo}t,\sigma$, at least one of the following is true:
            \begin{enumerate}[(a)]
                \item $\trace(s,\sigma)(1)\vDash\phi_0$ and $\trace(s,\sigma)(2)\vDash\phi_1$, or
                \item $\trace(s,\sigma)(1)\vDash\phi_1$ and $\trace(s,\sigma)(2)\vDash\phi_0$ and $\trace(s,\sigma)(3)\vDash\phi_2$, or
                \item $\trace(s,\sigma)(1)\vDash\phi_0$ and $\trace(s,\sigma)(2)\vDash\phi_2$ and $\trace(s,\sigma)(3)\vDash\phi_1$.
            \end{enumerate}
            And since $s,\sigma\sim_{\cor}r,\sigma$, at least one of the following is true:
            \begin{enumerate}[(a)]\setcounter{enumi}{3}
                \item $\trace(s,\sigma)(1)\vDash\widehat{\phi}^2_0$ and $\trace(s,\sigma)(2)\vDash\widehat{\phi}^2_1$, or
                \item $\trace(s,\sigma)(1)\vDash\widehat{\phi}^2_0$ and $\trace(s,\sigma)(3)\vDash\widehat{\phi}^2_2$, or
                \item $\trace(s,\sigma)(2)\vDash\widehat{\phi}^2_1$ and $\trace(s,\sigma)(3)\vDash\widehat{\phi}^2_2$.
            \end{enumerate}
            Now observe the following:
            \begin{itemize}
                \item If (d) is true then (b) is not (because (d) tells us $\trace(s,\sigma)(1)\vDash\neg\phi_1$ but (b) tells us $\trace(s,\sigma)(1)\vDash\phi_1$) and (c) is not (because (d) tells us $\trace(s,\sigma)(2)\vDash\neg\phi_2$ but (c) tells us $\trace(s,\sigma)(2)\vDash\phi_2$).
                \item If (e) is true then (b) is not (because (e) tells us $\trace(s,\sigma)(1)\vDash\neg\phi_1$ but (b) tells us $\trace(s,\sigma)(1)\vDash\phi_1$) and (c) is not (because (e) tells us $\trace(s,\sigma)(3)\vDash\neg\phi_1$ but (c) tells us $\trace(s,\sigma)(3)\vDash\phi_1$).
                \item If (f) is true then (b) is not (because (f) tells us $\trace(s,\sigma)(2)\vDash\neg\phi_0$ but (b) tells us $\trace(s,\sigma)(2)\vDash\phi_0$) and (c) is not (because (f) tells us $\trace(s,\sigma)(3)\vDash\neg\phi_1$ but (c) tells us $\trace(s,\sigma)(3)\vDash\phi_1$).
            \end{itemize}
            So (a) is the only possible option, and on (a), $s,\sigma\vDash\phi_0\land\XX\phi_1$, as required.
        \end{proof}

        Note that the antecedent of the theorem is not as unusual as it looks. As an example, a reading that fixes a particular value for an instrument will, by virtue of so fixing it, also guarantee that it doesn't take on any \emph{other} value---which is exactly what the $\langle\cor\rangle$-clause of the antecedent says. So a partial and not-quite-but-nearly correct rephrasing of the theorem would have it say that whenever you receive the same prediction for the next $n$ ticks from both a reordering channel and a corrupting channel, you can trust the first $n-1$ parts of the joint prediction. 

        There are a range of similar results that are valid in EDMon, each specifying a way of extracting more information by combining information from a dually monitored stream than can be captured from the two individual streams themselves. We will give a few instances of these to give a sense for the variety of validated inferences made possible by the move to EDMon, but will leave producing more such to the reader.

    \section{EDMon's NRV-adequacy}

        To demonstrate EDMon's usefulness as a formal framework for monitorability, we now turn to showing how to capture some of the central concepts in the NRV literature in EDMon. We begin with the monitorability itself in its various forms. We then turn to describing liveness classes, which are important to much of the classificatory work in the area. Finally, we turn to immunity which, as mentioned above, is a central aim of much work in the area. Each time, we show that many of the important concepts and many of the main results can be captured. 

    \subsection{Monitorability}

        A property is monitorable when a finite trace suffices to determine whether it holds. This loose intuition has been cashed out in a number of different ways. The discussion in \cite{kauffman2021can} identifies four varieties of monitorability: classical monitorability, $s$-monitorability, weak monitorability, and four-valued monitorability.\footnote{What we call $s$-monitorability is elsewhere called $\sigma$-monitorability. But this clashes with the metavariable conventions in this paper, so we've changed the name slightly.} We will restrict our attention to the first three of these since (a) this will take up enough space already and (b) it seems that EDMon is not the right framework for discussing four-valued monitorability and (c) four-valued monitorability is currently a minority position anyways.
    
        For our purposes, it suffices to think of each of classical monitorability, $s$-monitorability, and weak monitorability as describing a class of functions mapping pairs $\langle P,s\rangle$ where $P\subseteq\Sigma_S^\omega$ and $s\in\Sigma_S^*$ to a verdict in $\{\top, \bot, ?\}$. Intuitively, both the verdict $\top$ and the verdict $\bot$ mean that the finite sequence $s$ is sufficient to determine whether property $P$ holds, with the `$\top$' verdict meaning that $P$ determinately holds and the `$\bot$' verdict meaning that $P$ determinately doesn't hold. The meaning of the `$?$' verdict is (as we'll see) one of the areas of disagreement. 

    \subsubsection{$s$-monitorability}

        This form of monitorability was introduced in \cite{pnueli2006psl}. Philosophically, it differs from the other monitorability concepts we will discuss primarily in that it takes seriously the question of whether to continue monitoring. Intuitively, then, we might understand it reading the three verdicts as follows: `$\top$' is interpreted as `it is worth continuing to monitor because it's possible that $P$ will be positively determined by some further observation'; `$\bot$' is interpreted as `it is worth continuing to monitor because it's possible that $P$ will be negatively determined by some further observation'; and `$?$' is interpreted as `it is not worth continuing to monitor because no further observation will make either positive or a negative determination possible'. One thing this should make clear is that $s$-monitorability is done relative to a received string---that is, we make the judgment about whether continuing to monitor is worthwhile relative to what we've seen at some point. 

        Concretely, $s$-monitorability judgments are made via the following function:
        \begin{displaymath}
            \ev^{\smon}(P, s) = 
                \left\{
                    \begin{array}{ll}
                        \top & \text{if for some }t\in\Sigma_S^*, s . t . \sigma \in P\text{ for all }\sigma \in \Sigma_S^\omega \\
                        \bot & \text{if for some }t\in\Sigma_S^*, s . t . \sigma \not\in P\text{ for all }\sigma \in \Sigma_S^\omega \\
                        ? & \text{otherwise}
                    \end{array}
                \right.
        \end{displaymath}

        Note how this accomplishes the intended goal we described above: $\ev^{\smon}(P,s)=\top$ when there is something, $t$, that if we were to observe it, would put us in a position to say that, no matter how things turn out from there, $P$ will hold of the result. $\ev^{\smon}(P,s)=\bot$, similarly, holds if there is something that would let us say that no matter how things turn out, $P$ will fail to hold of the result. Both of these give us reason to continue monitoring. When they fail, we don't know whether $P$ will hold of the result---thus why the verdict is named `$?$'---but we do know it's not worth watching to find out because no finite amount of additional information will settle the fact. 
        
        One further note before moving on: we've written $s$-monitorability as a two-place function to facilitate comparison with the other types of monitorability we will introduce below. But there's a sense in which (a) $s$-monitorability is not a single monitorability concept but a family of monitorability concepts, one for each $s$ and (b) that, thought of in that way, each of them is in fact better understood as a \emph{unary} function $\ev^s(P)$, with $s$ as a parameter. Exactly one interesting thing hinges on this distinction (more on this below) but even if nothing did, it would nonetheless stand out as a difference worth marking.
            
        EDMon is equipped to model $s$-monitorability. To begin, we restrict our attention (as is usual) to the case where $P=\{\sigma\in\Sigma_S^\omega:\sigma\vDash_{\LTL} A\}$ for some $A\in\LL_{\LTL}$. Standardly, this set is called the language defined by $A$, and represented by $L\llbracket A\rrbracket$. To save on notation, we will usually write $\ev^{\smon}(A, s)$ rather than $\ev^{\smon}(L\llbracket A\rrbracket, s)$.
        
        \begin{theorem}
            $\ev^{\smon}(A, s) = \top$ iff $s, \overline{\emptyset} \vDash \Pp \langle\obs^*\rangle \Dd A$.
        \end{theorem}
        \begin{proof}
            Suppose $\ev^{\smon}(A, s)=\top$. Then for some $t\in\Sigma_S^*$, $s . t . \sigma \vDash_{\LTL} A$. So by Lemma~\ref{LTLSat} $s. t, \sigma \vDash A$ for all $\sigma \in \Sigma_S^\omega$. So $s . t, \overline{\emptyset} \vDash \Dd A$. But then by Lemma~\ref{obsklemma}, this happens iff $s, t . \overline{\emptyset} \vDash \langle\obs^*\rangle \Dd A$. But then by Lemma~\ref{DandPLemma}, $s, \overline{\emptyset} \vDash \Pp \langle \obs^* \rangle \Dd A$. 

            The reverse direction is equally straightforward and left to the reader. 
        \end{proof}
        Note that there is nothing special about $\overline{\emptyset}$ in the above proof. It simply makes a convenient `basepoint' as it were.

        \begin{theorem}
            $\ev^{\smon}(A, s) = \bot$ iff $s, \overline{\emptyset} \vDash \Pp \langle\obs^*\rangle \Dd \neg A$.
        \end{theorem}
        \begin{proof}
            By essentially the same reasoning as in the previous theorem. 
        \end{proof}
        
        If we say that $A$ is $s$-monitorable when either $\ev^{\smon}(A, s) = \top$ or $\ev^{\smon}(A, s) = \bot$, then the corresponding EDMon-concept is that $A$ is $s$-monitorable when $s, \overline{\emptyset} \vDash \Pp \langle\obs^*\rangle(\Dd A \lor \Dd \neg A)$.\smallskip

        \noindent\textit{Examples:} The choice of $s$ can matter a good deal in deciding whether $A$ is $s$-monitorable. As an example: $\XX p \lor \GG \FF p$ is $s$-monitorable for all length-1 sequences $s$, because for any such $s$ and all $\sigma\in\Sigma^\omega$, $s . \{p\}, \sigma \vDash \XX p \lor \GG \FF p$. But if $t \in \Sigma^*$ and $p \not\in t(2)$, then $\XX p \lor \GG \FF p$ is not $t$-monitorable because $t, \tau \vDash \XX p \lor \GG \FF p$ iff $t, \tau \vDash \GG \FF p$ and no finite prefix of $\trace(t, \tau)$ is sufficient to determine this. 

    \subsubsection{Classical Monitorability}

        Classical monitorability (per the discussion in \cite{bauer2011runtime}) was introduced as a generalization of $s$-monitorability. Of note is that we can capture the exact sense in which it succeeds to this end by comparing the two notions in the EDMon semantic framework; we will return to this below. 

        Mathematically, getting from $\ev^{\smon}$ to $\ev^{\cmon}$ is easy: all we have to do is drop the $t$'s. More to the point, and skipping straight to the properties we care about, we have the following:
        \begin{displaymath}
            \ev^{\cmon}(A, s) = 
                \left\{
                    \begin{array}{ll}
                        \top & s . \sigma \vDash_{\LTL} A \text{ for all }\sigma \in \Sigma_S^\omega \\
                        \bot & s . \sigma \vDash_{\LTL} \neg A \text{ for all }\sigma \in \Sigma_S^\omega \\
                        ? & \text{otherwise}
                    \end{array}
                \right.
        \end{displaymath}
        
        Note that this changes how we ought to interpret the verdicts $\top$, $\bot$, and $?$. The $s$-monitorability verdict $\top$ held when it was useful to continue monitoring because it remained possible to see something that would let us determine $A$. The classical monitorability verdict $\top$, on the other hand, makes the stronger claim that we have in fact determined $A$. (Similar things can obviously be said about $\bot$ and $?$.) 
            
        \begin{theorem}\label{pos_evc_theorem}
            $\ev^{\cmon}(A, s) = \top$ iff $s, \tau \vDash \Dd A$ for some $\tau\in\Sigma_S^\omega$.
        \end{theorem}
        \begin{proof}
            By Lemma~\ref{LTLSat}, $\ev^{\cmon}(A, s) = \top$ iff $s, \sigma \vDash A$ for all $\sigma \in \Sigma_S^\omega$. But this happens iff $s, \tau \vDash \Dd A$. 
        \end{proof}

        \begin{theorem}\label{neg_evc_theorem}
            $\ev^{\cmon}(A, s) = \bot$ iff $s, \tau \vDash \Dd \neg A$ for some $\tau\in\Sigma_S^\omega$.
        \end{theorem}
        \begin{proof}
            Essentially exactly as in the previous theorem.
        \end{proof}

        Classical monitorability is (unlike $s$-monitorability) genuinely a single monitorability concept. This has as a consequence that there is a univocal story about which properties are monitorable to be had, rather than the prefix-by-prefix story we saw in $s$-monitorability. To say more, we follow \cite{kupferman2001model} and say that $s$ is a good prefix for $A$ when $\ev^{\cmon}(A, s) = \top$ and say that $s$ is a bad prefix for $A$ when $\ev^{\cmon}(A, s) = \bot$. An ugly prefix for $A$ is prefix $s$ that can neither be extended to become a good prefix for $A$ nor extended to become a bad prefix for $A$---more formally, $s$ is an ugly prefix if for all $t \in \Sigma_S^*$, $\ev^{\cmon}(A, s . t)=?$. $A$ is classically monitorable when no prefix is ugly for $A$.

        \begin{theorem}
            $A$ is classically monitorable iff $\Pp \langle\obs^*\rangle (\Dd A \lor \Dd \neg A)$ is EDMon-valid iff $A$ is $s$-monitorable for all $s \in \Sigma_S^*$.
        \end{theorem} 
        \begin{proof}
            Suppose $A$ is classically monitorable and choose a point $s,\sigma$. Since $A$ is classically monitorable, either there is $t_1\in\Sigma_S^*$ so that $\ev^{\cmon}(A, s . t_1) = \top$ or there is $t_2\in\Sigma_S^*$ so that $\ev^{\cmon}(A, s . t_2) = \bot$. So by Theorems~\ref{pos_evc_theorem} and \ref{neg_evc_theorem}, either there is $t_1 \in \Sigma_S^*$ so that for some $\tau_1 \in \Sigma_S^\omega$, $s . t_1, \tau_1 \vDash \Dd A$ or there is $t_2 \in \Sigma_S^*$ so that for some $\tau_2 \in \Sigma_S^\omega$, $s . t_2, \tau_2 \vDash \Dd \neg A$. In either case, clearly $s . t_i, \tau_i \vDash \Dd A \lor \Dd \neg A$, and so by Lemma~\ref{obsklemma}, $s, t_i . \tau_i \vDash \langle\obs^*\rangle(\Dd A \lor \Dd \neg A)$. Thus, $s, \sigma \vDash \Pp \langle\obs^*\rangle(\Dd A \lor \Dd \neg A)$.

            Next, since $A$ is $s$-monitorable just if $s, \overline{\emptyset} \vDash \Pp \langle\obs^*\rangle (\Dd A \lor \Dd \neg A)$, clearly if $\Pp \langle\obs^*\rangle (\Dd A \lor \Dd \neg A)$ is valid, then for all $s$, $A$ is $s$-monitorable.

            Finally, suppose $A$ is $s$-monitorable for all $s$ and choose $s\in \Sigma_S^*$. Since $A$ is $s$-monitorable, $s, \overline{\emptyset} \vDash \Pp \langle\obs^*\rangle (\Dd A \lor \Dd \neg A)$. So by Lemma~\ref{DandPLemma} there is $\tau \in \Sigma_S^\omega$ so that $s, \tau \vDash \langle\obs^*\rangle (\Dd A \lor \Dd \neg A)$. Thus by Lemma~\ref{obsklemma}, for some $k$, $s . \tau_k, \tau^k \vDash \Dd A \lor \Dd \neg A$. It follows that either $s . \tau_k, \tau^k \vDash \Dd A$ or $s . \tau_k, \tau^k \vDash \Dd \neg A$. But either option guarantees that $s$ is not an ugly prefix, so $A$ is classically monitorable. 
        \end{proof}

        \noindent\textit{Examples:} We saw above that $\XX p\lor \GG \FF p$ is $s$-monitorable for some $s$ but not others. It follows from the previous theorem, then, that this formula is not classically monitorable. On the other hand, the formula $\FF p$ \emph{is} classically monitorable, since for any $s$ and any $\tau \in \Sigma^\omega$, $s, \{p\} . \tau \vDash \FF p$. 

        \subsubsection{Weak Monitorability} 
            
            As noted in \cite{kauffman2021can} there are properties that, even if they fail to be monitorable in any of the above senses, are nonetheless \emph{useful} to monitor. Weak monitorability, introduced in \cite{chen2018deciding,peled2019refining}, is aimed at exactly such cases. To say more, let's consider the example given in \cite{kauffman2021can}: $p \land \GG\FF p$. 
            \begin{theorem}
                If $s, \overline{\emptyset} \vDash p$, then $s, \overline{\emptyset} \not\vDash \Pp\langle\obs^*\rangle \Dd (p \land \GG \FF p)$ and $s, \overline{\emptyset} \not\vDash \Pp \langle\obs^*\rangle \Dd \neg (p \land \GG \FF p)$
            \end{theorem}
            \begin{proof}
                If $s=\varepsilon$, then since $s, \overline{\emptyset} \not\vDash p$, the theorem is trivially true. So suppose $s\neq\varepsilon$. By parsing semantic clauses we see that $s, \overline{\emptyset} \vDash \Pp \langle\obs^*\rangle \Dd (p \land \GG \FF p)$ iff there is $\sigma \in \Sigma_S^\omega$ and a number $k$ so that for all $\tau$, $s . \sigma_k, \tau \vDash p \land \GG \FF p$ while $s, \overline{\emptyset} \vDash \Pp \langle\obs^*\rangle \Dd \neg(p \land \GG \FF p)$ iff there is $\sigma \in \Sigma_S^\omega$ and a number $k$ so that for all $\tau$, $s . \sigma_k, \tau \vDash \neg(p \land \GG \FF p)$. Since $s \neq \varepsilon$, $\head(s, \overline{\emptyset}) = \head(s . \sigma_k, \tau)$. So since $s, \overline{\emptyset} \vDash p$, $s . \sigma_k, \tau \vDash p$. 

                So $s . \sigma_k, \tau \vDash p \land \GG \FF p$ iff $s . \sigma_k, \tau \vDash \GG \FF p$ and $s . \sigma_k, \tau \vDash \neg (p \land \GG \FF p)$ iff $s . \sigma_k, \tau \vDash \neg \GG \FF p$. But now observe that $s . \sigma_k, \overline{\{p\}} \vDash \GG \FF p$ while $s . \sigma_k, \overline{\emptyset} \vDash \neg \GG \FF p$. So it cannot be the case that for all $\tau$, $s . \sigma_k, \tau \vDash p \land \GG \FF p$ and it cannot be the case that for all $\tau$, $s . \sigma_k, \tau \vDash \neg (p \land \GG \FF p)$.
            \end{proof}

            It follows that $p \land \GG \FF p$ is not $s$-monitorable for any $s$ with $s, \overline{\emptyset} \vDash p$. In other words, every such $s$ is an ugly prefix for this formula. And since it admits ugly prefixes, it's also not classically monitorable. 
            
            But it's not as though it's completely useless to monitor $p \land \GG \FF p$---after all, $p \not\in \head(s, \sigma)$ immediately guarantees that $s, \sigma \not\vDash p \land \GG \FF p$. So there are \emph{some} prefixes that aren't ugly, and perhaps in some circumstances that gives us enough reason to monitor. Intuitively, weak monitorability is meant to be the kind of monitorability you should reach for when you have such reasons. Where classical monitorability demanded that \emph{no} prefix \emph{be} ugly, weak monitorability instead demands that \emph{some} prefix be \emph{non-}ugly. 
            \begin{theorem}
                $A$ is weakly monitorable iff $\Dd A \lor \Dd \neg A$ is satisfiable iff $A$ is $s$-monitorable for some $s$. 
            \end{theorem}
            \begin{proof}
                $A$ is weakly monitorable iff some $s$ is non-ugly iff there is an $s$ so that for some $t$, either $\ev^{\cmon}(A, s . t)=\top$ or $\ev^{\cmon}(A, s . t)=\bot$. By Theorems~\ref{pos_evc_theorem} and \ref{neg_evc_theorem}, this happens iff there is an $s$ so that for some $t$ either $s . t, \tau_1 \vDash \Dd A$ for some $\tau_1 \in \Sigma_S^\omega$ or $s . t, \tau_2 \vDash \Dd \neg A$ for some $\tau_2 \in \Sigma_S^\omega$, which happens iff there is an $s$ so that for some $t$ and some $\tau$, $s . t, \tau \vDash \Dd A \lor \Dd \neg A$ (which is to say, iff $\Dd A \lor \Dd \neg A$ is satisfiable). But this happens iff there is an $s$ so that for some $t$ and some $\tau$ $s, t . \tau \vDash \langle\obs^*\rangle (\Dd A \lor \Dd \neg A)$ which happens iff there is an $s$ so that $s, \overline{\emptyset} \vDash \Pp \langle\obs^*\rangle (\Dd A \lor \Dd \neg A)$ which is to say iff $A$ is $s$-monitorable for some $s$.
            \end{proof}

    \subsection{Liveness Classes}

        In \cite{peled2019refining}, Peled and Havelund introduce an extended version of Lamport's (see \cite{lamport1977proving}) safety-liveness hierarchy. In the extended hierarchy, there are six properties to track. Again, restricting to the case of properties defined by an LTL-formula, they can be characterized in our vocabulary as follows:
        \begin{itemize}
            \item $A \in \AFR$ (always finitely refutable) iff for all $\sigma \in \Sigma_S^\omega$, if $\sigma \not\vDash_{\LTL} A$, then for some $k$, $\sigma_k . \beta \not\vDash_{\LTL} A$ for all $\beta \in \Sigma_S^\omega$. 
            \item $A \in \AFS$ (always finitely satisfiable) iff for all $\sigma \in \Sigma_S^\omega$, if $\sigma \vDash_{\LTL} A$, then for some $k$, $\sigma_k . \beta \vDash_{\LTL} A$ for all $\beta \in \Sigma_S^\omega$.
            \item $A \in \SFR$ (sometimes finitely refutable) iff for some $\sigma \in \Sigma_S^\omega$, $\sigma \not\vDash_{\LTL} A$ and there is a $k$ so that $\sigma_k . \beta \not\vDash_{\LTL} A$ for all $\beta \in \Sigma_S^\omega$.
            \item $A \in \SFS$ (sometimes finitely satisfiable) iff for some $\sigma \in \Sigma_S^\omega$, $\sigma \vDash_{\LTL} A$ and there is a $k$ so that $\sigma_k . \beta \vDash_{\LTL} A$ for all $\beta \in \Sigma_S^\omega$.
            \item $A \in \NFR$ (never finitely refutable) iff for all $\sigma \in \Sigma_S^\omega$ and all $k$ there is $\beta$ so that $\sigma_k . \beta \vDash_{\LTL} A$
            \item $A \in \NFS$ (never finitely satisfiable) iff for all $\sigma \in \Sigma_S^\omega$ and all $k$ there is $\beta$ so that $\sigma_k . \beta \not\vDash_{\LTL} A$
        \end{itemize}

        The following theorem demonstrates that each of these classes is characterizable in EDMon:
        \begin{theorem}
            For all $A\in \LL_{\LTL}$, we have the following:
            \begin{itemize}
                \item $A \in \AFR$ iff $\neg A \to \langle\obs^*\rangle \Dd \neg A$ is EDMon-valid. 
                \item $A \in \AFS$ iff $A \to \langle\obs^*\rangle \Dd A$ is EDMon-valid. 
                \item $A \in \SFR$ iff $\neg A \land \langle\obs^*\rangle \Dd \neg A$ is EDMon-satisfiable. 
                \item $A \in \SFS$ iff $A \land \langle\obs^*\rangle \Dd A$ is EDMon-satisfiable. 
                \item $A \in \NFR$ iff $\Pp A$ is EDMon-valid
                \item $A \in \NFS$ iff $\Pp \neg A$ is EDMon-valid. 
            \end{itemize}
        \end{theorem}
        \begin{proof}
            For $\AFR$, observe that $\neg A \to \langle\obs^*\rangle \Dd \neg A$ is EDMon-valid iff for all $\langle s, \sigma\rangle$, if $s, \sigma \vDash \neg A$, then $s, \sigma \vDash \langle\obs^*\rangle \Dd \neg A$, which happens iff for all $\langle s, \sigma\rangle$, if $\trace(s, \sigma) \not\vDash_{\LTL} A$, then for some $k$ and all $\beta$, $\trace(s . \sigma_k, \beta) \not\vDash_{\LTL} A$. But observe that $\trace(s . \sigma_k, \beta) = \trace((s . \sigma)_k, (s . \sigma_k)^k . \beta)$, so this happens iff for all $\langle s, \sigma\rangle$, if $\trace(s, \sigma) \not\vDash_{\LTL} A$, then for some $k$ and all $\beta$, $\trace((s . \sigma)_k, (s . \sigma_k)^k. \beta) \not\vDash_{\LTL} A$. We take it to be clear that this happens iff $A \in \AFR$. The $\AFS$ case is essentially the same.

            For $\SFS$, note that $A \land \langle\obs^*\rangle \Dd A$ is EDMon-satisfiable iff $s, \sigma \vDash A \land \langle\obs^*\rangle \Dd A$ for some $\langle s, \sigma\rangle$, which happens iff for some $\langle s, \sigma\rangle$, $\trace(s, \sigma) \vDash A$ and for some $k$ and all $\beta$, $\trace(s . \sigma_k, \beta) \vDash A$. As in the previous case, this happens iff for some $\langle s, \sigma\rangle$, $\trace(s, \sigma) \vDash A$ and for some $k$ and all $\beta$, $\trace((s . \sigma)_k, (s . \sigma_k)^k . \beta) \vDash A$. Also as before, we take it to be clear that this happens iff $A \in \SFS$. The $\SFR$ case is essentially the same.

            Finally, for $\NFR$, note that $\Pp A$ is EDMon-valid iff for all $\langle s, \sigma\rangle$ there is $\beta$ so that $s, \beta \vDash A$ which happens iff for all $\langle s, \sigma\rangle$, there is $\beta$ so that $s . \beta \vDash_{\LTL} A$ which is essentially exactly what $\NFR$ demands. The $\NFS$ case is essentially the same. 
        \end{proof}

    \subsection{Further Property Classes}

        In \cite{kauffman2021can}, the authors introduced five further classes of properties and discussed their relations to the classes introduced above. We show here that all five of these can be easily represented in EDMon. 

        \subsubsection{The Classes}\label{classes_sec}

            The five classes introduced in \cite{kauffman2021can} are the classes of proximate, tolerant, permissive, inclusive, and exclusive properties. In the vocabulary of this paper (and again restricting to the formula-defined properties that are of interest in this paper) they are defined as follows:
            \begin{itemize}
                \item $A \in \Prox$ iff there is $s \in \Sigma_S^*$, $x \in \Sigma_S$ and $\sigma \in \Sigma_S^\omega$ so that either $s . x . \sigma \vDash_{\LTL} A$ and $s . x . x . \sigma \not\vDash_{\LTL} A$ or $s . x . \sigma \not\vDash_{\LTL} A$ and $s . x . x . \sigma \vDash_{\LTL} A$.
                \item $A \in \Tolr$ iff for all $\{s, t\} \subseteq \Sigma_S^*$ and all $\sigma \in \Sigma_S^\omega$, if $s . \sigma \vDash_{\LTL} A$, then $s . t . \sigma \vDash_{\LTL} A$.
                \item $A \in \Perm$ iff for all $\{s, t\} \subseteq \Sigma_S^*$ and all $\sigma \in \Sigma_S^\omega$, if $s . \sigma \not\vDash_{\LTL} A$, then $s . t . \sigma \not\vDash_{\LTL} A$.
                \item $A \in \Incl$ iff there is finite $T \subseteq \Sigma_S$ so that $\sigma \vDash_{\LTL} A$ iff each $t\in T$ occurs in $\sigma$.
                \item $A \in \Excl$ iff there is finite $T \subseteq \Sigma_S$ so that $\sigma \vDash_{\LTL} \neg A$ iff each $t\in T$ occurs in $\sigma$.
            \end{itemize}

            Before characterizing these classes in EDMon, it is useful to temporarily extend our vocabulary and allow a mutation we will call `blocked loss' and write `$\bloss$'. As with the other mutations considered so far, we interpret $\bloss$ by first describing its behavior on finite traces:
            \begin{itemize}
                \item $\widehat{\bloss} := \{\langle s, t \rangle \mid s = t$ or for some $\{u,x,v\} \subseteq \Sigma_S^*$, $s = u . x . v$ and $t = u . v\}$.
            \end{itemize}
        Observe that the difference between $\loss$ and $\bloss$ is in what they allow one to drop---in $\widehat{\loss}$, it's a single state; in $\widehat{\bloss}$ it's a finite sequence (block) of states. We interpret $\langle\bloss\rangle$ by appeal to $\llbracket\bloss\rrbracket := \Lift(\widehat{\bloss})$ just as we do with all the other mutations. As we leave to the reader to check, all the results proved above still carry over for the so-enriched language. 
        
        We now pause to address a somewhat delicate point. In \cite{kauffman2021can}, the authors say of the foursome $\ooo$, $\loss$, $\cor$, and $\st$ that `it is possible to apply a combination of these mutations to a trace to transform it into any other trace.' They formalize this in the following theorem that they describe as showing `Completeness of Mutations':
        \begin{quote}
            Given any two sets of non-empty [finite] traces $S,S'\subseteq\Sigma_S^*\setminus\{\varepsilon\}$, [there is] $k\in\mathbb{N}$ [such that] $(\loss\cup\cor\cup\st)^k = S \times S'$.
        \end{quote}
        As strictly written, this claim is false. Here are two straightforward ways to see this:
        \begin{itemize}        
            \item $\Sigma_S^*\setminus\{\varepsilon\}$ has uncountably many subsets. So there are uncountably many sets $S \times S'$ with $S,S'\subseteq\Sigma_S^*\setminus\{\varepsilon\}$. But there are only countably many sets of the form $(\loss\cup\cor\cup\st)^k$---one for each number $k$. So clearly it cannot be the case that all of the uncountably many sets $S\times S'$ are identical to some set of the form $(\loss\cup\cor\cup\st)^k$.
            \item Each member of $(\loss \cup \cor \cup \st)^k$ is infinite. If $S$ and $S'$ are both finite, then $S\times S'$ is finite. So for any such pair, it will not be the case that there is a $k$ so that $(\loss\cup\cor\cup\st)^k = S\times S'$.
        \end{itemize}
        Examining the proof they offer for this claim, one suspects the authors \emph{meant} to prove the following (true) claim that, in any event, is a better statement of the claim that combinations of $\ooo$, $\loss$, $\cor$, and $\st$ can transform any trace into any other trace:
        \begin{theorem}
            Given any two sets of non-empty finite traces $S,S'\subseteq\Sigma_S^*\setminus\{\varepsilon\}$, there is $k\in\mathbb{N}$ such that $(\loss\cup\cor\cup\st)^k \supseteq S \times S'$.
        \end{theorem}
        This claim is not only correct but can be proved along the exact lines of the proof they provide. However, it somewhat lacks the necessary power to be labeled a `completeness' result. After all, the set containing the single relation $\Sigma^* \times \Sigma^*$ has the same feature.%
        \footnote{Of course, there's nothing special about the particular threesome of $\loss$, $\cor$, and $\st$ here: given at most countably many generating relations and at most countably many ways of combining them, one can generate at most countably many relations.} The point, bringing it back to the matter of $\bloss$, is that $\bloss$ is one of the many relations inexpressible in terms of just $\loss$, $\cor$, $\st$, and $\ooo$. We omit a detailed proof of this fact, but note that in essence such a proof can be extracted from the observation that $\bloss$ does not move \emph{scattered} states but only sequential states and that it sometimes changes more than two states at a go. It's not hard to prove that no mutation generated by the usual foursome can have this combination of features. 
        
        Even so, $\bloss$ will be a temporary addition to our vocabulary for reasons that we will make clear in the next section of the paper. It's added, however, for its usefulness in stating the next result, before which we need one further definition: given a finite set of atoms $S$, an $S$-state description is a conjunction of the form $(\bigwedge_{t \in T} t) \land (\bigwedge_{t \in S \setminus T}\neg t)$ for some $T\subseteq S$. Each case of the following theorem is now essentially an immediate consequence of the definitions of the terms involved:
            
        \begin{theorem}\label{five_char}
            For all $A \in \LL_\LTL$ we have the following:
            \begin{itemize}
                \item $A \in \Prox$ iff either $(A \land \langle\st\rangle \neg A)$ or $(\neg A \land \langle\st\rangle A)$ is EDMon-satisfiable. 
                \item $A \in \Tolr$ iff $A \to [\bloss^-] A$ is EDMon-valid. 
                \item $A \in \Perm$ iff $\neg A \to [\bloss^-]\neg A$ is EDMon-valid. 
                \item For finite $S\subseteq \At$, $A\in\Incl$ iff there are state descriptions $c_1, \dots, c_n$ so that $A\otto \bigwedge_{i=1}^n \FF c_i$ is $S$-valid.
                \item For finite $S\subseteq \At$, $A\in\Excl$ iff there are state descriptions $c_1, \dots, c_n$ so that $\neg A \otto \bigwedge_{i=1}^n \FF c_i$ is $S$-valid.
            \end{itemize}
        \end{theorem}
        In the latter two cases we say that $A$ (resp. $\neg A$) is equivalent to a conjunction of $\FF$-ed state descriptions. 

            \subsubsection{Results about the Classes}

            The following inclusions were claimed in \cite{kauffman2021can}:
            \begin{itemize}
                \item $\Tolr \subseteq \NFR$. 
                \item $\Perm \subseteq \NFS$.
                \item $\Incl \subseteq \Tolr \setminus \NFS$
                \item $\Excl \subseteq \Perm \setminus \NFR$.
            \end{itemize}
            The first two are not quite correct---the empty property ($\bot$) is tolerant, but obviously not in $\NFR$; the trivial property ($\top$) is permissive, but obviously not in $\NFS$. But once we correct for that, the results hold, as we leave the reader to check. The EDMon-validities that correspond to the corrected inclusions are the following:

            \begin{theorem}
                For all $A \in \LL_\LTL$ we have the following:
                \begin{itemize}
                    \item If $A \to [\bloss^-]A$ is EDMon-valid then $A \to \Pp A$ is EDMon-valid.
                    \item If $\neg A \to [\bloss^-]\neg A$ is EDMon-valid then $\neg A \to \Pp \neg A$ is EDMon-valid.
                    \item If $A$ is equivalent to a conjunction of $\FF$-ed state descriptions, then $A\to[\bloss^-]A$ is EDMon-valid and $\Pp \neg A$ is not EDMon-valid.
                    \item If $\neg A$ is equivalent to a conjunction of $\FF$-ed state descriptions, then $\neg A\to[\bloss^-]\neg A$ is EDMon-valid and $\Pp A$ is not EDMon-valid.
                \end{itemize}
            \end{theorem}
            \begin{proof}
                We show the first and the fourth and leave the others to the reader. For the first, suppose $A\to[\bloss^-]A$ is EDMon-valid and let $s, \sigma \vDash A$. Since $A \in \LL_\LTL$, it follows that $\varepsilon, s . \sigma \vDash A$. So $\varepsilon, s . \sigma \vDash [\bloss^-]A$. But now observe that $s \sim_{\widehat{\bloss}} \varepsilon$, so $\varepsilon \sim_{\widehat{\bloss^-}} s$. So $s, s. \sigma \vDash A$, and thus $s, \sigma \vDash \Pp A$. 

                For the fourth, suppose that $\neg A$ is equivalent to $\bigwedge_{i=1}^n \FF c_i$. If $s, \sigma \vDash \neg A$, then there are (not necessarily distinct) numbers $m_1, \dots, m_n$ so that $\shift^{m_i}(s, \sigma) \vDash c_i$ for $1\leq i\leq n$. Let $s \sim_{\widehat{\bloss^-}} t$. Then either $s=t$ or there are $\{u, x, v\} \subseteq \Sigma^*_S$ so that $s = u . v$ and $t = u . x . v$. Observe that for $m_i < |u|$, that $\head(\shift^{m_i}(s, \sigma)) = \head(\shift^{m_i}(t, \sigma))$ and that for $m_i \geq |u|$, $\head(\shift^{m_i}(s, \sigma)) = \head(\shift^{m_i + |x|}(t, \sigma))$. So there are (not necessarily distinct) numbers $m'_1, \dots, m'_n$ so that $\shift^{m'_i}(t, \sigma) \vDash c_i$ for $1\leq i\leq n$ and thus $t, \sigma \vDash \neg A$. So $s, \sigma \vDash [\bloss^-] \neg A$. Thus $\neg A \to [\bloss^-]\neg A$ is valid. 

                To see that $\Pp A$ is not EDMon-valid, let $S_i$ be the set containing all positive literals in $c_i$. Then $ S_1\dots S_n , \tau \vDash \neg A$ for all $\tau$. So for no $\tau$ does $S_1\dots S_n , \tau \vDash A$ and thus $ S_1\dots S_n , \overline{\emptyset} \not\vDash \Pp A$
            \end{proof}
        
    \subsection{Immunity}

        At last we come to immunity, which has been of central importance in NRV. For the forms of monitorability and kinds of properties that are of interest here, immunity to mutation is defined in \cite{kauffman2021can} as follows:
        $A$ is true-false immune to $\mu$ iff $\ev^{\cmon}(A, s) = \ev^{\cmon}(A, t)$ for all $\langle s, t\rangle \in \widehat{\mu}$. In EDMon, we can begin to capture this in the following way:

        \begin{lemma}
            $A$ is true-false immune to $\mu$ iff all four of the following are true:
            \begin{itemize}
                \item for all $s \sim_{\widehat{\mu}} t$, if $s, \sigma \vDash \Dd A$ for some $\sigma$ then $t, \tau \vDash \Dd A$ for some $\tau$, and 
                \item for all $s \sim_{\widehat{\mu}} t$, if $t, \sigma \vDash \Dd A$ for some $\sigma$ then $s, \tau \vDash \Dd A$ for some $\tau$, and
                \item for all $s \sim_{\widehat{\mu}} t$, if $s, \sigma \vDash \Dd \neg A$ for some $\sigma$ then $t, \tau \vDash \Dd \neg A$ for some $\tau$, and
                \item for all $s \sim_{\widehat{\mu}} t$, if $t, \sigma \vDash \Dd \neg A$ for some $\sigma$ then $s, \tau \vDash \Dd \neg A$ for some $\tau$.
            \end{itemize}
        \end{lemma}
        \begin{proof}
            $A$ is true false-immune to $\mu$ iff $\ev^{\cmon}(A, s) = \ev^{\cmon}(A, t)$ for all $\langle s, t\rangle \in \widehat{\mu}$. But this happens iff for all $s \sim_{\widehat{\mu}} t$, we have all of the following:
            \begin{itemize}
                \item if $\ev^{\cmon}(A, s) = \top$, then $\ev^{\cmon}(A, t) = \top$, and
                \item if $\ev^{\cmon}(A, t) = \top$, then $\ev^{\cmon}(A, s) = \top$, and
                \item if $\ev^{\cmon}(A, s) = \bot$, then $\ev^{\cmon}(A, t) = \bot$, and
                \item if $\ev^{\cmon}(A, t) = \bot$, then $\ev^{\cmon}(A, s) = \bot$, and
            \end{itemize}
            After distributing the universal and applying Theorems~\ref{pos_evc_theorem} and \ref{neg_evc_theorem}, the result is immediate. 
        \end{proof}

        With one more lemma, we can characterize immunity in a perspicuous way:

        \begin{lemma}
            $s, \sigma \vDash \Dd A$ iff there is some $\tau$ so that $s, \tau \vDash \Dd A$. 
        \end{lemma}
        \begin{proof}
            Left to right is immediate. For right to left, suppose $s, \tau \vDash \Dd A$. Then $s, \tau' \vDash \Dd A$ for all $\tau'$. But then $s, \sigma \vDash \Dd A$ as well. 
        \end{proof}
        
        \begin{theorem}\label{tfi}
            $A$ is true-false immune to $\mu$ iff all of the following are EDMon-valid:
            \begin{itemize}
                \item $\Dd A \to [\mu] \Dd A$
                \item $\Dd A \to [\mu^-] \Dd A$
                \item $\Dd \neg A \to [\mu] \Dd \neg A$
                \item $\Dd \neg A \to [\mu^-] \Dd \neg A$
            \end{itemize}
        \end{theorem}
        \begin{proof}
            Note that $\Dd A \to [\mu] \Dd A$ is valid iff for all $\langle s, \sigma\rangle$, if $s, \sigma \vDash \Dd A$, then $s, \sigma \vDash [\mu] \Dd A$. By the previous lemma, this happens iff for all $s \sim_{\widehat{\mu}} t$, if $s, \sigma \vDash \Dd A$ for some $\sigma$, then $t, \tau \vDash \Dd A$ for some $\tau$. A similar argument shows that $\Dd \neg A \to [\mu] \Dd \neg A$ is valid iff for all $s \sim_{\widehat{\mu}} t$, if $s, \sigma \vDash \Dd \neg A$ for some $\sigma$, then $t, \tau \vDash \Dd \neg A$ for some $\tau$. And a similar pair of arguments establish that $\Dd A \to [\mu^-] \Dd A$ is valid iff for all $s \sim_{\widehat{\mu}} t$, if $t, \sigma \vDash \Dd A$ for some $\sigma$ then $s, \tau \vDash \Dd A$ for some $\tau$ and $\Dd \neg A \to [\mu^-] \Dd \neg A$ is valid iff for all $s \sim_{\widehat{\mu}} t$, if $t, \sigma \vDash \Dd \neg A$ for some $\sigma$ then $s, \tau \vDash \Dd \neg A$ for some $\tau$.
        \end{proof}

        Also isolated in \cite{kauffman2021can} is the restricted form of immunity known as \emph{trustworthiness}. The definition they provide is this: a verdict $v\in\{\top,\bot,?\}$ is trustworthy for $A$ with respect to $\mu$ when for all $s\sim_\mu t$, if $\ev^{\cmon}(A,t)=v$ then $\ev^{\cmon}(A,s)=v$. Two thirds of the following theorem characterizing trustworthiness are then essentially immediate: 

        \begin{theorem}\label{twtheorem}
            For all $A\in\LL_{\LTL}$ and all mutations $\mu$, $\top$ is trustworthy for $A$ with respect to $\mu$ iff $\Dd A \to [\mu^-] \Dd A$ is EDMon-valid; $\bot$ is trustworthy for $A$ with respect to $\mu$ iff $\Dd \neg A \to [\mu^-] \Dd \neg A$ is EDMon-valid; and $?$ is trustworthy for $A$ with respect to $\mu$ iff $(\Pp A \land \Pp \neg A) \to [\mu^-](\Pp A \land \Pp\neg A)$ is valid.
        \end{theorem}
        \begin{proof}
            Only the $?$-part needs work. So, suppose that for all $s \sim_{\mu} t$, if $\ev^{\cmon}(A,t)=?$, then $\ev^{\cmon}(A,s)=?$. Then for all $s \sim_{\mu} t$ if for all $\tau$, $t, \tau \not\vDash \Dd A$ and for all $\tau$, $t, \tau \not\vDash \Dd \neg A$, then for all $\sigma$, $s, \sigma \not\vDash \Dd A$ and $s, \sigma \not\vDash \Dd \neg A$. But this happens iff for all $s \sim_{\mu} t$, if $t, \tau \vDash \Pp A \land \Pp\neg A$, then $s, \tau \vDash \Pp A \land \Pp\neg A$. But this happens iff $t, \tau \vDash (\Pp A \land \Pp \neg A) \to [\mu^-] (\Pp A \land \Pp \neg A)$ for all $\langle t, \tau\rangle$, which is to say iff $(\Pp A \land \Pp \neg A) \to [\mu^-](\Pp A \land \Pp\neg A)$ is valid.
        \end{proof}

        Note that the antecedents of these conditions are exhaustive: clearly for all points $\langle s, \sigma\rangle$, either $s, \sigma \vDash \Dd A \lor \Dd\neg A$ or $s, \sigma \vDash \neg(\Dd A \lor \Dd\neg A)$. But in the latter case, $s, \sigma \vDash \Pp A \land \Pp \neg A$. Using this, we then get an analogue of Corollary 2 from \cite{kauffman2021can}:
        \begin{coro}\label{tfitw}
            $A$ is true-false immune to $\mu$ iff $\top$, $\bot$, and $?$ are all trustworthy for $A$ with respect to $\mu$.
        \end{coro}
        \begin{proof}
            By Theorem~\ref{tfi} it suffices to show that $\top$ , $\bot$, and $?$ are all trustworthy for $A$ with respect to $\mu$ iff all four of $\Dd A \to [\mu] \Dd A$, $\Dd A \to [\mu^-] \Dd A$, $\Dd \neg A \to [\mu] \Dd \neg A$, and $\Dd \neg A \to [\mu^-] \Dd \neg A$ are valid. But $\top$ is trustworthy iff $\Dd A \to [\mu^-] \Dd A$ is valid and $\bot$ is trustworthy iff $\Dd \neg A \to [\mu^-] \Dd \neg A$. And if $\Dd A \to [\mu] \Dd A$ and $\Dd \neg A \to [\mu] \Dd \neg A$ are both valid then it's not too hard to see that $(\Pp A \land \Pp \neg A) \to [\mu^-](\Pp A \land \Pp \neg A)$ is also valid.

            So what remains is showing that if $\Dd A \to [\mu^-] \Dd A$ and $\Dd \neg A \to [\mu^-] \Dd \neg A$ and $(\Pp A \land \Pp \neg A) \to [\mu^-](\Pp A \land \Pp \neg A)$ are all valid, then so too are $\Dd A \to [\mu] \Dd A$ and $\Dd \neg A \to [\mu] \Dd \neg A$. 

            To that end, suppose $s, \sigma \vDash \Dd A$ and $s \sim_{\widehat{\mu}} t$. To show that $t, \sigma \vDash \Dd A$ it suffices (by the observation preceding this result) to show that $t, \sigma \not\vDash \Dd \neg A$ and $t, \sigma \not\vDash \Pp A\land \Pp \neg A$. But this isn't hard: since $s \sim_{\widehat{\mu}} t$, $t \sim_{\widehat{\mu^-}} s$. Since $\Dd \neg A \to [\mu^-]\Dd \neg A$ is valid, if $t, \sigma \vDash \Dd \neg A$, then $s, \sigma \vDash \Dd \neg A$, contradicting $s, \sigma \vDash \Dd A$. Similarly, if $t, \sigma \vDash \Pp A \land \Pp \neg A$, then $s, \sigma \vDash \Pp \neg A$, again a contradiction. So $t, \sigma \vDash \Dd A$. A similar argument shows that $\Dd \neg A \to [\mu] \Dd \neg A$.
        \end{proof}

    \section{Some EDMon Validities} 

        Giving a complete axiomatization of the set of EDMon-valid formulas is probably impossible. But we can still make headway on describing some of the important classes of EDMon-validities, and we'll end our introduction of EDMon by doing so, and then showing a few things the resulting partial axiomatization can be useful for.  

        \subsection{Relationship With Other Logics}
    
        We begin by describing some of the important closure principles EDMon obeys. Here is the first, most basic one whose proof we omit:
        \begin{theorem}\label{taut}
            All EDMon-instances of propositional tautologies are EDMon-valid.
        \end{theorem}
        The proof of each part of the following theorem is trivial, so omitted:
        \begin{theorem}\label{rules}
            All of the following rules preserve EDMon-validity for all EDMon-formulas $A$ and $B$, LTL-formulas $\phi$ and mutations $\mu$:
            \AxiomC{$A$} \AxiomC{$B$} \BinaryInfC{$A \land B$} \DisplayProof,
            \AxiomC{$A \to B$} \AxiomC{$A$} \BinaryInfC{$B$} \DisplayProof,
            \AxiomC{$A$} \UnaryInfC{$[\mu] A$} \DisplayProof,
            \AxiomC{$A$} \UnaryInfC{$[\obs] A$} \DisplayProof,
            \AxiomC{$A$} \UnaryInfC{$\Dd A$} \DisplayProof,
            \AxiomC{$\phi$} \UnaryInfC{$\XX \phi$} \DisplayProof, and
            \AxiomC{$\phi$} \UnaryInfC{$\GG \phi$} \DisplayProof
        \end{theorem}

        \begin{lemma}
            For $\phi\in\LL^S_{\mathrm{LTL}}$, if $\phi$ is LTL-valid, then $\phi$ is $S$-valid. 
        \end{lemma}
        \begin{proof}
            Suppose $\phi\in\LL_{\mathrm{LTL}}$ is not EDMon-valid. Then there are $s\in\Sigma_S^*,\sigma\in\Sigma_S^\omega$ such that $s,\sigma\not\vDash\phi$. 
            By Lemma \ref{LTLSat}, $\trace(s,\sigma)\not\vDash_{\LTL}\phi$. Therefore, $\phi$ is not LTL-valid.
        \end{proof}

        \begin{coro}
            For all $\phi_1$, $\phi_2$, and $\phi_3$ in $\LL_\LTL$, all of the following are valid:\footnote{This is the axiomatization presented in \cite[\S4.3]{sep-logic-temporal}.}
            \begin{multicols}{2}\begin{enumerate}[({LTL}1)]
                \item $\GG (\phi_1 \to \phi_2) \to (\GG \phi_1 \to \GG \phi_2)$
                \item $\XX (\phi_1 \to \phi_2) \to (\XX \phi_1 \to \XX \phi_2)$
                \item $\XX \neg \phi_1 \otto \neg \XX \phi_1$
                \item $\GG \phi_1 \otto (\phi_1 \land \XX \GG \phi_1)$
                \item $(\phi_2 \land \GG (\phi_2 \to (\phi_1 \land \XX \phi_2))) \to \GG \phi_1$
                \item $(\phi_1 \UU \phi_2) \otto (\phi_2 \lor (\phi_1 \land \XX (\phi_1 \UU \phi_2)))$
                \item $\GG ((\phi_2 \lor (\phi_1 \land \XX \phi_3)) \to \phi_3) \to ((\phi_1 \UU \phi_2) \to \phi_3)$
            \end{enumerate}\end{multicols}
        \end{coro}
        
        Recall that test-free propositional dynamic logic with converse (PDL$^-$; see \cite{harel2000dynamic}) has the same mutation-combining (or, in PDL's vocabulary, program-combining) operations as EDMon, but a countable infinity of atomic programs and no `LTL-sublayer'. It is axiomatized as the modus-ponens and necessitation closure of propositional logic plus the following:
        \begin{multicols}{2}\begin{enumerate}[({PDL}1)]
            \item $[\alpha] (A \to B) \to ([\alpha] A \to [\alpha] B)$
            \item $[\alpha] (A \land B) \otto ([\alpha] A \land [\alpha] B)$
            \item $[\alpha_1 \cup \alpha_2] A \otto ([\alpha_1] A \land [\alpha_2] A)$
            \item $[\alpha_1; \alpha_2] A \otto [\alpha_1] [\alpha_2] A$
            \item $(A \land [\alpha] [\alpha^*] A) \otto [\alpha^*] A$
            \item $[\alpha^*] (A \to [\alpha] A) \to (A \to [\alpha^*] A)$
            \item $A \to [\alpha] \langle\alpha^-\rangle A$
            \item $A \to [\alpha^-] \langle\alpha\rangle A$
        \end{enumerate}\end{multicols}  

        \begin{lemma}
            All well-formed EDMon-substitutions of PDL-theorems are EDMon-valid.
        \end{lemma}

        Before proving this lemma, we'll note that there's a bit of sleight of hand going on here with `well-formed substitutions'. The point is that in the axioms for PDL, one can substitute any mutation for $\alpha$, can substitute both $\obs$ and $\obs^*$ for $\alpha$, and can substitute $\Dd$ whenever the result is grammatical and is a genuine substitution. Concretely, this means that all of these substitutions are permitted in (PDL1) and (PDL2), that only mutations are allowed in (PDL3), (PDL4), (PDL7) and (PDL8), and that mutations and observations are both allowed in PDL5 and PDL6. 
        \begin{proof}
            Given Theorems~\ref{taut} and \ref{rules}, it suffices to show that every well-formed instance of the above axioms is EDMon-valid. We provide the proofs for three of them and leave the rest to the reader.

            For PDL1, we consider three cases: $\alpha=\mu$ a mutation, $\alpha=\obs$, and $\alpha=\Dd$. For the first case, suppose $s, \sigma \vDash [\mu] (A \to B)$. To see that $s, \sigma \vDash [\mu] A \to [\mu] B$, suppose $s, \sigma \vDash [\mu] A$ and let $s \sim_{\widehat{\mu}} t$. Then $t, \sigma \vDash A \to B$ and $t, \sigma \vDash A$. So $t, \sigma \vDash B$ and thus $s, \sigma \vDash [\mu] A \to [\mu] B$. 

            For the second case, suppose $s, \sigma \vDash [\obs] (A \to B)$ and let $s, \sigma \vDash [\obs] A$. Then $s . \sigma_1, \sigma^1 \vDash A \to B$ and $s . \sigma_1, \sigma^1 \vDash A$. So $s . \sigma_1, \sigma^1 \vDash B$.

            For the third case, suppose $s, \sigma \vDash \Dd (A \to B)$ and let $s, \sigma \vDash \Dd A$. Choose $\tau \in \Sigma^\omega$. Then $s, \tau \vDash A \to B$ and $s, \tau \vDash A$. So $s, \tau \vDash B$.

            For PDL6, we consider: $\alpha = \mu$ a mutation and $\alpha = \obs$. But as was just witnessed in the PDL1 case, the two end up being essentially parallel so we consider only the first. Thus, suppose $s, \sigma \vDash [\mu^*] (A \to [\mu] A)$ and $s, \sigma \vDash A$. We show by induction that for all $k$, if $s \sim_{\widehat{\mu^k}} t$, then $t, \sigma \vDash A$. Since $s, \sigma \vDash A$, the base case is already established. Now suppose that for all $s \sim_{\widehat{\mu^k}} t$, $t, \sigma \vDash A$. Since $s, \sigma \vDash [\mu^*] (A \to [\mu] A)$, $t, \sigma \vDash A \to [\mu] A$. So since $t, \sigma \vDash A$, $t, \sigma \vDash [\mu] A$. So for all $u$ with $t \sim_{\widehat{\mu}} u$, $u, \sigma \vDash A$. Thus for all $u$ with $s \sim_{\widehat{\mu^{k+1}}} u$, $u, \sigma \vDash A$.

            For PDL8, the only well-formed options are the $\alpha = \mu$ options. So, suppose $s, \sigma \vDash A$ and choose $t$ so that $s \sim_{\mu^-} t$. Then $t \sim_{\mu} s$. So since $s, \sigma \vDash A$, $t, \sigma \vDash \langle \mu \rangle A$. Thus since $t$ was arbitrary, $s, \sigma \vDash [\mu^-]\langle\mu\rangle A$.
        \end{proof}

        \begin{coro}\label{somePDL}
            The following are EDMon-valid: $[(\mu^*)^-]A\otto[(\mu^-)^*]A$, $[(\mu_1 \cup \mu_2)^-]A \otto [\mu_1^-\cup\mu_2^-]A$, and $[(\mu_1; \mu_2)^-]A \otto [\mu_2^-; \mu_1^-]A$.
        \end{coro}

        \subsection{General Modal Features}

        \begin{theorem}
            $\dd$ is reflexive, as is $\llbracket\mu\rrbracket$ for all mutations $\mu$.
        \end{theorem}

        \begin{proof} 
            Immediate for $\dd$; by a straightforward induction on mutations for $\mu$.
        \end{proof}

        \begin{coro}
            All instances of $[\mu]A\to A$ and $\Dd A \to A$ are valid.
        \end{coro}

        \begin{theorem}\label{mutationFacts}
            All of the following hold:
            \begin{multicols}{2}\begin{enumerate}[(a)]
                \item $\widehat{\st^{-}} \subseteq \widehat\loss$,  $\widehat\st \subseteq \widehat{\loss^{-}}$
                \item $\widehat{\ooo^k} \subseteq \widehat{\cor^{2k}}$
                \item $\widehat\ooo = \widehat{\ooo^{-}}$, $\widehat\cor = \widehat{\cor^{-}}$
                \item $\widehat{\cor^k} \subseteq (\widehat\loss; \widehat{\loss^{-}})^k$
            \end{enumerate}\end{multicols}
        \end{theorem}
        
        \begin{proof} 
        For part (a), let $\langle s,t\rangle\in\widehat{\st^{-}}$. Therefore, there are $u,v\in\Sigma_S^*$ and $x\in\Sigma_S$ such that $s=uxxv$ and $t=uxv$. By definition, $\langle s,t\rangle\in\widehat\loss$. The second part is similar.  

        Part (b) is proved by induction on $k$. The inductive case is handled by the following observation. Suppose $\langle s,t\rangle\in\widehat\ooo$, so there are $u,v\in\Sigma_S^*$ and $x,y\in\Sigma_S$ such that $s=uxyv$ and $t=uyxv$. Then, $\langle uxyv,uxxv\rangle\in\widehat\cor$, and $\langle uxxv,uyxv\rangle\in\widehat\cor$, so $\langle s,t\rangle
        \in\widehat{\cor}^2$.

        Part (c) is immediate from the definitions. 

        Part (d) is proved by induction on $k$. The inductive case is handled by the following observation. Suppose $\langle s,t\rangle\in\widehat\cor$, where for some $u,v\in\Sigma_S^*$ and $x,y\in\Sigma_S$, $s=uxv$ and $t=uyv$. Then $\langle uxv, uv\rangle\in\widehat{\loss}$, $\langle uyv, uv\rangle\in\widehat{\loss}$, and $\langle uv, uyv\rangle\in\widehat{\loss^{-}}$, so 
        $\langle s, t\rangle\in(\widehat\loss; \widehat{\loss^{-}})$. 
%
        \end{proof}

        \begin{coro}
            All instances of each of the following are valid:
            \begin{multicols}{2}\begin{enumerate}[(a)]
                \item $[\loss] A \to [\st^-] A$; $[\loss^-] A \to [\st] A$
                \item $[\cor^{2k}] A \to [\ooo^k] A$
                \item $[\ooo] A \otto [\ooo^-] A$; $[\cor] A \otto [\cor^-] A$
                \item $[(\loss;\loss^-)^k] A \to [\cor^k] A$
            \end{enumerate}\end{multicols}
        \end{coro}

Next, we will turn to some results building to a result concerning the relation $\llbracket(\loss \cup \cor \cup \st)^*\rrbracket$, which we discussed in \S\ref{classes_sec}: 
        \begin{lemma}\label{Lemma:LossRelation}
        For all $s\in\Sigma_S^*$, $\langle s,\varepsilon\rangle\in\widehat{\loss}^*$
        \end{lemma}
        \begin{proof}
        The proof is by induction on $k=|s|$. For $k=0$, $\langle\varepsilon,\varepsilon\rangle\in\widehat{\loss}$, by definition. 

        For the inductive case, we have $s=tx$, for some $t\in\Sigma_S^*$ and $x\in\Sigma_S$. By the inductive hypothesis, $\langle t,\varepsilon\rangle\in\widehat{\loss}^*$, and $\langle tx,x\rangle\in\widehat{\loss}$, by definition. Therefore, $\langle s,\varepsilon\rangle\in\widehat{\loss}^*$, as desired.
        \end{proof}

        \begin{lemma}\label{tdrop}
        For all $s,t\in\Sigma_S^*$, $\langle st,s\rangle\in\widehat{\loss}^*$
        \end{lemma}
        \begin{proof}
        The proof is by induction on $|t|$. If $|t|=0$, then the result is immediate.

        Suppose $t=ux$, for some $u\in\Sigma_S^*$ and $x\in\Sigma_S$. 
        By the inductive hypothesis, 
        $\langle su, s\rangle\in\widehat{\loss}^*$. As 
        $\langle sux,su\rangle\in\widehat{\loss}^*$, it follows that 
        $\langle st, s\rangle\in\widehat{\loss}^*$, as desired.
    
        \end{proof}
        
       \begin{lemma}\label{Lemma:CorruptingRel}
        For all $s,t\in\Sigma_S^*$, if $|s|=|t|$, then $\langle s,t\rangle\in\widehat{\cor}^*$
        \end{lemma}
        \begin{proof}
        The proof is by induction on $k=|s|$. For $k=0$, $\langle\varepsilon,\varepsilon\rangle\in\widehat{\cor}$, by definition. 

        For the inductive case, let $s=ux$ and $t=vy$, for some $x,y\in\Sigma_S$ and some $v,u\in\Sigma_S^*$ such that $|u|=|v|$. By the inductive hypothesis, 
        $\langle u,v\rangle\in\widehat{\cor}^*$, so 
        $\langle ux,vx\rangle\in\widehat{\cor}^*$. It then follows that $\langle ux,vy\rangle\in\widehat{\cor}^*$, as desired. 
        \end{proof}

        \begin{lemma}\label{Lemma:SequenceExtend}
        Let $s\in\Sigma_S^*$ be such that $|s|\geq 1$. For all $t\in\Sigma_S^*$, $\langle s, st\rangle\in(\widehat{\st}\cup\widehat{\cor})^*$.
        \end{lemma}
        \begin{proof}
        Let $s\in\Sigma_S^*$ be such that $|s|\geq 1$. The proof is by induction on $|t|$. The case where $|t|=0$ is immediate.

        Suppose $|t|=k+1$, so $t=ux$, for some $u\in\Sigma_S^*$ such that $|u|=k$ and $x\in\Sigma_S$. If $k\geq 1$, then $u=vy$, for some $v\in\Sigma_S^*$ and $y\in \Sigma_S$.
        By the inductive hypothesis, $\langle s,svy\rangle\in(\widehat{\st}\cup\widehat{\cor})^*$. 
        It follows that $\langle s,svyy\rangle\in(\widehat{\st}\cup\widehat{\cor})^*$, so $\langle s,svyx\rangle\in(\widehat{\st}\cup\widehat{\cor})^*$, as desired.
        Suppose that $k=0$. Then since  $|s|\geq 1$, $s=wz$, for some $w\in\Sigma_S^*$ and $z\in\Sigma_S$. It follows that
        $\langle wz,wzz\rangle\in(\widehat{\st}\cup\widehat{\cor})^*$, so
        $\langle wz,wzx\rangle\in(\widehat{\st}\cup\widehat{\cor})^*$, which is 
        $\langle s,sx\rangle\in(\widehat{\st}\cup\widehat{\cor})^*$, as desired.
        \end{proof}

        \begin{lemma}\label{Lemma:AlmostUniversal}
        Let $s,t\in\Sigma_S^*$ be such that $|s|\geq 1$ and $|t|\geq 1$. Then, 
        $\langle s, t\rangle\in(\widehat{\st}\cup\widehat{\cor}\cup\widehat{\loss})^*$.
        \end{lemma}
        \begin{proof}
        Let $s,t\in\Sigma_S^*$ be such that $|s|\geq 1$ and $|t|\geq 1$. Suppose $|s|\leq|t|$. 
        For some $u,v\in\Sigma_S^*$, $t=uv$, where $|u|=|s|$. 
        By Lemma \ref{Lemma:CorruptingRel}, 
        $\langle s, u\rangle\in(\widehat{\st}\cup\widehat{\cor}\cup\widehat{\loss})^*$. By Lemma \ref{Lemma:SequenceExtend}, 
        $\langle u, uv\rangle\in(\widehat{\st}\cup\widehat{\cor}\cup\widehat{\loss})^*$, so 
        $\langle s, uv\rangle\in(\widehat{\st}\cup\widehat{\cor}\cup\widehat{\loss})^*$, which is to say
        $\langle s, t\rangle\in(\widehat{\st}\cup\widehat{\cor}\cup\widehat{\loss})^*$, as desired.

        Suppose $|s|>|t|$. Then $s=uv$, for some $u,v\in\Sigma_S^*$ where $|u|=|t|$. By Lemma \ref{tdrop}, 
        $\langle uv, u\rangle\in(\widehat{\st}\cup\widehat{\cor}\cup\widehat{\loss})^*$. By Lemma \ref{Lemma:CorruptingRel},
        $\langle u, t\rangle\in(\widehat{\st}\cup\widehat{\cor}\cup\widehat{\loss})^*$, so 
        $\langle uv, t\rangle\in(\widehat{\st}\cup\widehat{\cor}\cup\widehat{\loss})^*$, which is to say
        $\langle s, t\rangle\in(\widehat{\st}\cup\widehat{\cor}\cup\widehat{\loss})^*$, which was to be proved. 
        \end{proof}
Let $\Sigma_S^+=\set{s\in\Sigma_S^*:|s|\geq 1}$. Then, $(\widehat{\st}\cup\widehat{\cor}\cup\widehat{\loss})^*$ is the universal relation on $\Sigma_S^+$. It is not the universal relation on $\Sigma_S^*$, in fact it is not even an equivalence relation, since symmetry can fail. For every $s\in\Sigma_S^*$, $\langle s,\varepsilon\rangle\in(\widehat{\st}\cup\widehat{\cor}\cup\widehat{\loss})^*$, but when $s\neq\varepsilon$,
$\langle \varepsilon,s\rangle\not\in(\widehat{\st}\cup\widehat{\cor}\cup\widehat{\loss})^*$. If instead of $\widehat{\st}$, we used the relation that simply inserted an element, rather than duplicating an existing element, then we would have the universal relation, rather than something merely close to it. Let us make this more precise. 
        \begin{theorem}\label{Theorem:AlmostUnivReln}
           For all $s,t\in\Sigma_S^*$, if $|s|\geq 1$, then $\langle s,t\rangle \in(\widehat{\st}  \cup \widehat{\cor} \cup \widehat{\loss})^*$
        \end{theorem}
        \begin{proof}
        Suppose $s,t\in\Sigma_S^*$ and $|s|\geq 1$. If $|t|=0$, then by Lemma \ref{Lemma:LossRelation}, 
        $\langle s,t\rangle\in(\widehat{\st}\cup\widehat{\cor}\cup\widehat{\loss})^*$. If $|t|\geq 1$, then by Lemma \ref{Lemma:AlmostUniversal}, 
        $\langle s,t\rangle\in(\widehat{\st}\cup\widehat{\cor}\cup\widehat{\loss})^*$.
        \end{proof}

        As discussed above, the mutation $(\st\cup\cor\cup\loss)^*$ plays a distinguished role in the work of \cite{kauffman2021can}. We have (slightly) generalized their result, and below we will show that the modal version of this mutation has {S5}-like axioms. First, we will need an auxiliary result:

        \begin{lemma}
            For $k\geq 0$, $(\Dd \XX^k p \lor \Dd \XX^k\neg p)$ is true at exactly the points $\langle s,\sigma\rangle$ with $|s|>k$. 
        \end{lemma}
            \begin{proof}
                Moving quickly through the semantics, note that $s, \sigma \vDash \Dd \XX^k p \lor \Dd \XX^k \neg p$ iff either for all $\tau$, $\shift^k(s, \tau) \vDash p$ or for all $\tau$, $\shift^k(s, \tau)\vDash \neg p$. 

                If $|s|>k$, then $\head(\shift^k(s, \tau)) = s(k+1)$. So if $p\in s(k+1)$, then for all $\tau$, $\shift^k(s, \tau) \vDash p$ while if $p\not\in s(k+1)$ for all $\tau$, $\shift^k(s, \tau) \vDash \neg p$. So $s, \sigma \vDash \Dd \XX^k p \lor \Dd \XX^k \neg p$. 

                If $|s| = l \leq k$, then $\head(\shift^k(s,\tau)) = \tau(k-l+1)$. So if $\tau = \overline{\emptyset}$, then $\shift^k(s, \tau) \not\vDash p$ while if $\tau=\overline{\{p\}}$, $\shift^k(s, \tau) \not\vDash \neg p$. So it is not the case that all $\tau$, $\shift^k(s, \tau) \vDash p$ and it is not the case that all $\tau$, $\shift^k(s, \tau) \vDash \neg p$. So $s, \sigma \not\vDash \Dd \XX^k p \lor \Dd \XX^k \neg p$.
            \end{proof}

        It follows, of course, that $\neg(\Dd \XX^k p \lor \Dd \XX^k\neg p)$ is true at exactly the points $\langle s,\sigma\rangle$ with $|s|\leq k$. In particular, $\neg(\Dd p \lor \Dd\neg p)$ is true at exactly initial points, so we call it $\init$.

        \begin{theorem}
            All instances of the following are EDMon-valid:
            \begin{displaymath}
                \langle(\loss\cup\cor\cup\st)^*\rangle A\to [(\loss\cup\cor\cup\st)^*](\init \lor \langle(\loss\cup\cor\cup\st)^*\rangle A)
            \end{displaymath} 
        \end{theorem}
        \begin{proof}
            We prove this via reductio. Assume $s,\sigma \vDash \langle(\loss\cup\cor\cup\st)^*\rangle A$ and 
            $s,\sigma \not\vDash [(\loss\cup\cor\cup\st)^*](\init \lor \langle(\loss\cup\cor\cup\st)^*\rangle A)$. So, there is $t$ such that $s \sim_{\widehat{(\loss\cup\cor\cup\st)^*}} t$ so that $t, \sigma \not\vDash \init$ and $t, \sigma \not\vDash \langle(\loss \cup \cor \cup \st)^*\rangle A$. Since $t, \sigma \not\vDash \init$, $|t| \neq 0$. Since $t, \sigma \not\vDash \langle(\loss \cup \cor \cup \st)^*\rangle A$, for all $u$, if $t \sim_{\widehat{(\loss\cup\cor\cup\st)^*}} u$, then $u, \sigma \not\vDash A$. 
            
            Since $|t|\neq 0$, it follows by Theorem~\ref{Theorem:AlmostUnivReln} that $u, \sigma \not\vDash A$ for all $u$. But since $s,\sigma \vDash \langle(\loss\cup\cor\cup\st)^*\rangle A$, there is some $v$ so that $v, \sigma \vDash A$, which is a contradiction.
        \end{proof}

        Thus, we see that the mutation $(\loss\cup\cor\cup\st)^*$ has S5-like axioms, here having demonstrated the validity of a 5-like axiom.
        Above we highlighted a difference that would emerge had we focused on a different mutation, namely an insertion rather than $\st$. Due to the richness of the language, we have this mutation available, it is $\loss^-$. This suggests that the mutation $(\loss\cup\loss^-)^*$ would be worth a look for comparison. 
        
        \begin{lemma}
        $\widehat{(\loss\cup\cor\cup\st)^*} \subseteq\widehat{(\loss\cup\loss^-)^*}$.
        \end{lemma}
            \begin{proof}
            By Theorem \ref{mutationFacts}, $\widehat{\st}\subseteq\widehat{\loss^-}$  and $\widehat{\cor}^*\subseteq\widehat{(\loss;\loss^-)^*}\subseteq\widehat{(\loss\cup\loss^-)^*}$. Therefore,         $\widehat{(\loss\cup\cor\cup\st)^*} \subseteq\widehat{(\loss\cup\loss^-)^*}$, as desired.       
            \end{proof}
        \begin{coro}\label{Corollary:LossAlmostUniversal}
        For all $s,t\in\Sigma_S^*$, if $|s|\geq 1$, then $\langle s,t\rangle \in\widehat{(\loss\cup\loss^-)^*}$
        \end{coro}
        \begin{proof}
        By Theorem \ref{Theorem:AlmostUnivReln}  and the preceding lemma. 
        \end{proof}
        \begin{lemma}\label{Lemma:LossZerotoOne}
        For all $s\in\Sigma^*_S$ such that $|s|=1$, $\langle\varepsilon,s\rangle\in\widehat{(\loss\cup\loss^-)^*}$
        \end{lemma}
        \begin{proof}
        For all $s\in \Sigma^*_S$ such that $|s|=1$, $\langle s,\varepsilon\rangle\in\widehat{\loss}$. Therefore, by definition, 
        $\langle \varepsilon,s\rangle\in\widehat{\loss^-}$, so 
        $\langle \varepsilon,s\rangle\in\widehat{(\loss\cup\loss^-)^*}$, as desired. 
        \end{proof}
        \begin{theorem}\label{Theorem:UnivRelLoss}
        For all $s,t\in\Sigma_S^*$,  $\langle s,t\rangle \in\widehat{(\loss\cup\loss^-)^*}$
        \end{theorem}
        \begin{proof}
        This is immediate from Lemma \ref{Lemma:LossZerotoOne} and Corollary \ref{Corollary:LossAlmostUniversal}.
        \end{proof}
        Thus, we establish that  $\widehat{(\loss\cup\loss^-)^*}$ is universal on $\Sigma^*_S$. As a consequence, it is immediate that the corresponding mutation obeys S5 principles.
        \begin{theorem}
        All instances of the following are EDMon-valid:
            \begin{displaymath}
                \langle(\loss\cup\loss^-)^*\rangle A\to [(\loss\cup\loss^-)^*] \langle(\loss\cup\loss^-)^*\rangle A
            \end{displaymath}     
        \end{theorem}
        \begin{proof}
        This follows from the fact that $\widehat{((\loss\cup\loss^-)^*)}$ is universal on $\Sigma^*_S$ as in Theorem \ref{Theorem:UnivRelLoss}.
        \end{proof}
        The logic of the $(\loss\cup\loss^-)^*$ mutation is, then, properly S5, unlike the logic of ${(\loss\cup\cor\cup\st)^*}$. Using EDMon, we have highlighted a small difference in mutations that results in a substantive logical difference. We think this further demonstrates the utility of EDMon, and suggests further investigation into the different combinations of mutations would be worthwhile. 

        \subsection{Specialized Validities}

        So far we've examined \emph{general} EDMon validities---validities that hold because of EDMon's relationship to other logics. Now we look at a few EDMon validities that hold because of its own peculiar features. In fact, we've already seen a few such results: 
        \begin{itemize}
            \item For $\phi\in\LL_{\LTL}$, $\phi\to\langle\obs^*\rangle \phi$ is EDMon-valid (Lemma \ref{LTLobsLemma})
            \item For all $A$ and all mutations $\mu$, $\langle\mu\rangle\langle\obs\rangle A\to \langle\obs\rangle\langle\mu\rangle A$ is EDMon-valid (Lemma~\ref{muobslemma}).
            \item For all $A$, $\Dd\langle\obs\rangle A\to\langle\obs\rangle\Dd A$ is EDMon-valid (Lemma~\ref{obsCommutativity}).
            \item For all $A$ and all mutations $\mu$, $\Dd[\mu]A\otto[\mu]\Dd A$ is EDMon-valid (Lemma~\ref{mucomm}). 
        \end{itemize}

        These are misleadingly simple. They have the form of axioms one is accustomed to seeing in other areas of modal logic. To get a sense for how much of a genuine zoo the set of EDMon validities is, the reader should remind themselves of the content of Theorems \ref{FFtheorem} through \ref{multichan}. Or, consider the following theorem:
        \begin{theorem}
            For all $A$ and $B$ in $\LL_{\LTL}$, the following is EDMon-valid:
            \begin{displaymath}
                (\Pp\langle\obs^*\rangle\Dd A \land \Pp\langle\obs^*\rangle\Dd B) \to \Pp\langle\obs^*\rangle(\Dd A\land\langle\cor^*\rangle\Dd B)
            \end{displaymath}
        \end{theorem}
        \begin{proof}
            Suppose $s, \sigma\vDash\Pp\langle\obs^*\rangle\Dd A$ and $s, \sigma\vDash\Pp\langle \obs^* \rangle \Dd B$. Then there are $\tau$, $\rho$, $k$, and $l$ so that $s . \tau_{k},\tau^{k} \vDash \Dd A$ and $s . \rho_{l},\rho^{l} \vDash \Dd B$. It follows that $s . \tau_{k},\beta \vDash A$ for all $\beta$ and $s . \rho_{l},\beta \vDash B$ for all $\beta$. So in particular, $s . \tau_{k},(\tau_{\max(k,l)})^{k} . \alpha \vDash A$ for all $\alpha$ and $s . \rho_{l},(\rho_{\max(k,l)})^{l} . \alpha \vDash B$ for all $\alpha$. Since $A$ and $B$ are in $\LL_{\LTL}$, it follows from Lemma~\ref{LTLSat} that $s . \tau_{k} .(\tau_{\max(k,l)})^{k} , \alpha \vDash A$ for all $\alpha$ and $s . \rho_{l} . (\rho_{\max(k,l)})^{l} , \alpha \vDash B$ for all $\alpha$. So $s . \tau_{\max(k,l)},\tau^{\max(k,l)} \vDash \Dd A$ and $s. \rho_{\max(k,l)},\tau^{\max(k,l)}\vDash \Dd B$. Next note that $s. \tau_{\max(k,l)} \sim_{\widehat{\cor^{\max(k,l)}}} s. \rho_{\max(k,l)}$. So since $s. \rho_{\max(k,l)},\tau^{\max(k,l)}\vDash \Dd B$, we have that $s . \tau_{\max(k,l)},\tau^{\max(k,l)} \vDash \langle\cor^*\rangle \Dd B$. So $s,\tau\vDash\langle\obs^*\rangle(\Dd A\land \langle\cor^*\rangle\Dd B)$. Thus $s, \sigma \vDash \Pp\langle\obs^*\rangle(\Dd A\land\langle\cor^*\rangle\Dd B)$.
        \end{proof}
        Note that this is, modally speaking, unusual behavior. We have here a way to factor paired diamond-shaped modalities out of a \emph{conjunction}, provided the result is mediated by a third diamond-shaped modality.

        \subsection{Putting the Validities to Work}

        \begin{figure}
            \begin{multicols}{2}
            \begin{enumerate}[({SV}1)]\setlength\itemsep{1ex}
                \item $\phi\to\langle\obs^*\rangle\phi$
                \item $\langle\mu\rangle \langle\obs\rangle  A \to \langle\obs\rangle \langle\mu\rangle A$
                \item $\Dd\langle\obs\rangle A\to \langle\obs\rangle\Dd A$
                \item $\Dd[\mu]A\otto[\mu]\Dd A$
            \end{enumerate}
            \begin{enumerate}[({LTL}1)]\setlength\itemsep{1ex}
                \item $\GG (\phi_1 \to \phi_2) \to (\GG \phi_1 \to \GG \phi_2)$
                \item $\XX (\phi_1 \to \phi_2) \to (\XX \phi_1 \to \XX \phi_2)$
                \item $\XX \neg \phi_1 \otto \neg \XX \phi_1$
                \item $\GG \phi_1 \otto (\phi_1 \land \XX \GG \phi_1)$
                \item $(\phi_2 \land \GG (\phi_2 \to (\phi_1 \land \XX \phi_2))) \to \GG \phi_1$
                \item $(\phi_1 \UU \phi_2) \otto (\phi_2 \lor (\phi_1 \land \XX (\phi_1 \UU \phi_2)))$
                \item $\GG ((\phi_2 \lor (\phi_1 \land \XX \phi_3)) \to \phi_3) \to ((\phi_1 \UU \phi_2) \to \phi_3)$
            \end{enumerate}
            \begin{enumerate}[({PDL}1)]\setlength\itemsep{1ex}
                \item $[\alpha] (A \to B) \to ([\alpha] A \to [\alpha] B)$
                \item $[\alpha] (A \land B) \otto ([\alpha] A \land [\alpha] B)$
                \item $[\alpha_1 \cup \alpha_2] A \otto ([\alpha_1] A \land [\alpha_2] A)$
                \item $[\alpha_1; \alpha_2] A \otto     [\alpha_1] [\alpha_2] A$
                \item $(A \land [\alpha] [\alpha^*] A) \otto [\alpha^*] A$
                \item $[\alpha^*] (A \to [\alpha] A) \to (A \to [\alpha^*] A)$
                \item $A \to [\alpha] \langle\alpha^-\rangle A$
                \item $A \to [\alpha^-] \langle\alpha\rangle A$
            \end{enumerate}
            \begin{enumerate}[({GM}1)]\setlength\itemsep{1ex}
                \item $[\mu] A \to A$
                \item $\Dd A \to A$
                \item $[\loss] A \to [\st^-] A$; $[\loss^-] A \to [\st] A$
                \item $[\cor^{2k}] A \to [\ooo^k] A$
                \item $[\ooo] A \otto [\ooo^-] A$; $[\cor] A \otto [\cor^-] A$
                \item $[(\loss;\loss^-)^k] A \to [\cor^k] A$
                \item $\langle(\loss \cup \cor \cup\st)^*\rangle A \to [(\loss \cup \cor \cup \st)^*](\init \lor \langle(\loss \cup \cor \cup \st)^*\rangle A)$
                \item $\langle(\loss\cup\loss^-)^*\rangle A\to [(\loss\cup\loss^-)^*] \langle(\loss\cup\loss^-)^*\rangle A$
            \end{enumerate}           
            \begin{enumerate}[({R}1)]\setlength\itemsep{1ex}
                \item \AxiomC{$A$} \AxiomC{$B$} \BinaryInfC{$A \land B$} \DisplayProof
                \item \AxiomC{$A \to B$} \AxiomC{$A$} \BinaryInfC{$B$} \DisplayProof
                \item \AxiomC{$A$} \UnaryInfC{$[\mu] A$} \DisplayProof
                \item \AxiomC{$A$} \UnaryInfC{$[\obs] A$} \DisplayProof
                \item \AxiomC{$A$} \UnaryInfC{$\Dd A$} \DisplayProof
                \item \AxiomC{$\phi$} \UnaryInfC{$\XX \phi$} \DisplayProof
                \item \AxiomC{$\phi$} \UnaryInfC{$\GG \phi$} \DisplayProof
            \end{enumerate}\end{multicols}
            \caption{A Few Axioms and Rules for EDMon. ($\phi$ and $\phi_i$ in $\LL_{\LTL}$)}\label{alltheaxioms}
        \end{figure}

        In Figure~\ref{alltheaxioms} we collect the axioms and rules we have shown EDMon to accept. In spite of EDMon-validity's zoo-like behavior that was just demonstrated, we would like to show that the particular family singled out here is nonetheless quite often useful. We will do so by verifying EDMon-analogues of a few results in the literature. Here is a nice example of such: Theorem 2 of \cite{kauffman2021can} says, in our vocabulary, that $A$ is true-false immune to $\mu$ iff $A$ is true-false immune to $\mu^*$. We can prove this as follows:
        \begin{theorem}\label{tfistar}
            $A$ is true-false immune to $\mu$ iff $A$ is true false immune to $\mu^*$.
        \end{theorem}
        \begin{proof}
            The `if' direction is immediate. For the `only if' direction, note by Theorem~\ref{tfi} that if $A$ is true-false immune to $\mu$, then $\Dd A \to [\mu] \Dd A$, $\Dd A \to [\mu^-] \Dd A$, $\Dd \neg A \to [\mu] \Dd \neg A$, and $\Dd \neg A \to [\mu^-] \Dd \neg A$ are all valid. It follows by R3 that $[\mu^*](\Dd A \to [\mu] \Dd A)$, $[(\mu^-)^*](\Dd A \to [\mu^-] \Dd A)$, $[\mu^*](\Dd \neg A \to [\mu] \Dd \neg A)$, and $[(\mu^-)^*](\Dd \neg A \to [\mu^-] \Dd \neg A)$ are all valid. So by R2, PDL6, and Corollary~\ref{somePDL}, $\Dd A \to [\mu^*] \Dd A$, $\Dd A \to [(\mu^*)^-] \Dd A$, $\Dd \neg A \to [\mu^*] \Dd \neg A$, and $\Dd \neg A \to [(\mu^*)^-] \Dd \neg A$ are all valid. So again by Theorem~\ref{tfi}, $A$ is true false immune to $\mu^*$. 
        \end{proof}

        Of course, a version of this result for trustworthy verdicts (which would correspond to Corollary 1 in \cite{kauffman2021can}) is immediately available by essentially the same proof:
        \begin{theorem}\label{twstar}
            If $v\in\{\top,\bot,?\}$ is trustworthy for $A$ with respect to $\mu$, then $v$ is trustworthy for $A$ with respect to $\mu^*$.
        \end{theorem}

        We can also use the axioms just introduced to explain our previous insistence on $\bloss$ being a merely temporary addition to our vocabulary. Recall that we introduced $\bloss$ in order to express $\Tolr$ in EDMon-vocabulary: $A\in\Tolr$ iff $A\to[\bloss^-]A$ is EDMon valid. But by the proof of Theorem~\ref{tfistar}, we see that $A\to[\bloss^-]A$ is EDMon-valid iff $A\to[(\bloss^-)^*]A$ is EDMon-valid. But also, $\widehat{(\loss^-)^*}=\widehat{(\bloss^-)^*}$ because scattered insertions (what $(\loss^-)^*$ does) just are scattered block insertions of block-length 1 (which $(\bloss^-)^*$ can do) and $n$ scattered block insertions of lengths $k_1, \dots, k_n$ (the result of any single application of $(\bloss^*)$) can be accomplished by $\sum_{i=1}^n k_i$ applications of $\loss^-$, which $(\loss^-)^*$ allows. 
        
        So $A\to[(\bloss^-)^*]A$ is valid iff $A\to[(\loss^-)^*]A$ is valid. It follows that $A\to[\bloss^-]A$ is valid iff $A\to[(\loss^-)^*]A$ is valid iff $A\to[\loss^-]A$ is valid. Since the same chain of argumentation clearly works to show that $\neg A\to[\bloss^-]\neg A$ iff $\neg A\to[(\loss^-)^*]\neg A$, $\bloss$ can be dispensed with in the only two roles it played: characterizing $\Tolr$ and characterizing $\Perm$.

        Moving on, we provide a different demonstration of the value of the axioms. In \cite{kauffman2021can}, they state the following as their Theorems 8 and 9 (paraphrased and restricted to the cases of concern):
        \begin{itemize}
            \item $A\in\Tolr$ iff $\top$ is trustworthy for $A$ with respect to $\loss$.
            \item $A\in\Perm$ iff $\bot$ is trustworthy for $A$ with respect to $\loss$.
        \end{itemize}
        For each of these, only one direction of the biconditional is actually proved in \cite{kauffman2021can}. One suspects that the other direction was meant to be removed in editing, since in both cases the other direction is in fact false:
        \begin{theorem}
            $\top$ is (vacuously) trustworthy for $\GG p$, but $\GG p$ is not tolerant. $\bot$ is (vacuously) trustworthy for $\FF p$ but $\FF p$ is not permissive.
        \end{theorem}
        \begin{proof}
            Since $\neg \Dd \GG p$ is valid, so too is $\Dd \GG p\to[\mu^-]\GG p$. So by Theorem~\ref{twtheorem}, $\top$ is trustworthy for $\GG p$. But clearly $\GG p$ is not tolerant, as insertion of any single non-$p$ state will falsify it. A similar argument works for the other half.
        \end{proof}

        However, for the direction of each of these results that \emph{does} hold, the above axiom system admits a very quick proof:
        \begin{theorem}
            If $A\in\Tolr$, then $\top$ is trustworthy for $A$ with respect to $\loss$ and if $A\in\Perm$ then $\bot$ is trustworthy for $A$ with respect to $\loss$. 
        \end{theorem}
        \begin{proof}
            By Theorem~\ref{five_char} and the discussion above, if $A\in\Tolr$, then $A\to[\loss^-]A$ is EDMon-valid. Thus by R5 so is $\Dd(A\to[\loss^-] A)$. It follows by R2 and PDL1 that $\Dd A \to \Dd[\loss^-] A$ is valid. Thus by SV4 so is $\Dd A \to [\loss^-]\Dd A$. It follows by Theorem~\ref{twtheorem} that $\top$ is trustworthy for $A$ with respect to $\loss$. A similar argument works in the other case. 
        \end{proof}

        We should also be clear that these validities are quite useful even in cases where they can't do the job entirely on their own. As an example, we have the following theorem, which is a version of Theorem 11 in \cite{kauffman2021can} and in which we blend syntactic and semantic reasoning:
        \begin{theorem}
            For $A\in\LL_{\LTL}$, if $\Dd A\to[\cor]\Dd A$ and $\Dd\neg A\to[\cor]\Dd \neg A$ are EDMon-valid, then exactly one of $\Dd A$, $\Dd\neg A$ and $\Pp A\land \Pp\neg A$ is as well.
        \end{theorem}
        \begin{proof}
            Since $(\Dd A \lor \Dd\neg A)\lor(\Pp A\land\Pp\neg A)$ is equivalent to an instance of excluded middle, at every point at least one of $\Dd A$, $\Dd\neg A$ and $\Pp A\land \Pp\neg A$ is true. It's also clear that at every point, no more than one of them can be true, from which it follows that (a) no more than one of them can be valid and (b) at every point exactly one of them is true. 
            
            To see that at least one is, first note that if $\Dd A\to[\cor]\Dd A$ and $\Dd\neg A \to[\cor]\Dd \neg A$ are both EDMon-valid, then by GM5, $\Dd A\to[\cor^-]\Dd A$ and $\Dd\neg A\to[\cor^-]\Dd\neg A$ are as well. Thus by Theorem~\ref{tfi}, $A$ is true-false immune to $\cor$. So by Theorem~\ref{tfistar}, $A$ is true-false immune to $\cor^*$. So by Corollary~\ref{tfitw} and Theorem~\ref{twtheorem} $\Dd A\to[\cor^*]\Dd A$ and $\Dd\neg A\to[\cor^*]\Dd\neg A$ are EDMon-valid.

            Suppose at some point $\langle s,\sigma\rangle$, $\Dd A$ is true. Then $s, \sigma \vDash [\cor^*]\Dd A$. But for all $t$ with $|t|=|s|$, $s\sim_{\widehat{\cor^*}} t$. So $t, \tau \vDash A$ for all $t$ with $|t|=|s|$ and all $\tau$. Thus by Lemma~\ref{LTLSat}, $\rho\vDash_{\LTL} A$ for all $\rho\in\Sigma^\omega$, from which we obviously get that $\Dd A$ is EDMon-valid. So the result follows if $\Dd A$ is ever true. It follows similarly if $\Dd\neg A$ is ever true, in which case the same argument gives that $\Dd\neg A$ is EDMon-valid. And if $\Dd A$ and $\Dd\neg A$ are never true, then $\Dd A\lor \Dd\neg A$ is never true. So $\neg(\Dd A \lor \Dd\neg A)$ is always true. But then $\Pp A\land \Pp\neg A$ is always true.
        \end{proof}

    \section{Conclusion}

        It's hard to compare things. EDMon makes it easier. There are lots of options. EDMon makes them visible. It's easy to make mistakes. EDMon helps correct them, as has been shown above. But there remain mistakes outside of the scope of EDMon left to be caught. So in conclusion we reflect on a direction in which generalizing EDMon is likely to extend the utility of the approach even more broadly. Consider the work of Finkbeiner on \emph{hyperproperties} in \cite{finkbeiner-2023}, that is, relations that hold not of between individual traces but between sets of traces. 
        
        While the formalism introduced in this paper is insufficiently expressive to articulate many hyperproperties, the scheme may clearly be adapted to model monitorability not only of LTL but of the more expressive logics in which hyperproperties may be addressed. For example, the logic HyperLTL introduced in \cite{clarkson-2014} can be used as a basis on which the additional mutation modalities may be added.

        To sketch such a generalization in a way that meets our own idioms, we note that the syntax of HyperLTL includes quantifiers $\exists \pi$ and $\forall\pi$ where $\exists \pi.A$ and $\forall \pi.A$ are read as ``there exists a trace such that $A$'' and ``for any trace, $A$,'' respectively. To more formally express this, first consider a set of \emph{trace variables} $\mathcal{V}=\lbrace\pi_{1},\pi_{2},...,\rbrace$ and a \emph{trace assignment} as a function $\Pi:\mathcal{V}\to\Sigma_S^* \times \Sigma_S^\omega$. Let a trace substitution $\Pi[\pi\to\lbrace s,\sigma\rangle]$ be the assignment differing from $\Pi$ only by mapping instances of the variable $\pi$ to $\lbrace s,\sigma\rangle$. Validity can be then defined in terms of a trace assignments with $T\subseteq\Sigma_S^* \times \Sigma_S^\omega$ by extending a natural reinterpretation of the validity clauses to hold between a formula and a tuple $T,\Pi,s,\sigma$. Then include the following:

        \begin{itemize}[leftmargin=*]\setlength\itemsep{1ex}
            \item $T,\Pi,s, \sigma \vDash \exists \pi.  A$ iff there exists a $\langle s,\sigma\rangle\in T$ such that $T,\Pi[\pi\to\langle s,\sigma\rangle],s,\sigma\vDash A$
            \item $T,\Pi,s, \sigma \vDash \for all \pi.  A$ iff for all $\langle s,\sigma\rangle\in T$, it holds that $T,\Pi[\pi\to\langle s,\sigma\rangle],s,\sigma\vDash A$
        \end{itemize}

        As Finkbeiner catalogs, many hyperproperties can be represented in the language of HyperLTL that express important security-related concepts. These include observational determinism (asserting the indistinguishability of traces based on low-security input, see \cite{zdancewic-2003}) or noninference (asserting that arbitrary replacement of high-security inputs will not be reflected in observations, see \cite{mclean-1994}). Coupling logics like HyperLTL with the techniques described above to define \emph{e.g.} a HyperEDMon promises to provide novel ways to consider an even broader palette of considerations from data security.

    \bibliography{biblio}
    \bibliographystyle{ACM-Reference-Format}
    
\end{document}